\documentclass[onecolumn]{IEEEtran}
\IEEEoverridecommandlockouts
\usepackage[final]{graphicx}
\usepackage{cite}
\usepackage{amssymb}
\usepackage{amsmath}
\usepackage{amsthm}
\usepackage{amsfonts}
\usepackage{bm}
\usepackage{float}
\usepackage{filecontents}
\usepackage{algorithm}
\usepackage{algpseudocode}
\usepackage{algorithmicx}
\usepackage{subfigure}
\usepackage{subfigmat}

\title{Symbol-Level Precoding Design Based on Distance Preserving Constructive Interference Regions}

\author{Alireza~Haqiqatnejad,~\emph{Student Member,~IEEE},~Farbod~Kayhan,~and Bj\"{o}rn~Ottersten,~\emph{Fellow,~IEEE}
	\thanks{The authors are with Interdisciplinary Centre for Security, Reliability and Trust (SnT), University of Luxembourg, L-1855 Luxembourg. (email:~\{alireza.haqiqatnejad,farbod.kayhan,bjorn.ottersten\}@uni.lu).}
	\thanks{Part of this work has been submitted to the 19\textsuperscript{th} IEEE International Workshop on Signal Processing Advances in Wireless Communications (SPAWC), Kalamata, 2018.}
}

\newtheorem{theorem}{Theorem}
\newtheorem{lemma}{Lemma}
\newtheorem{property}{Property}

\newtheorem{proposition}{Proposition}

\newcommand{\Deee} {\mathrm{\boldsymbol{\delta}}}

\newcommand{\Gammaaa} {\mathrm{\mathbf{\Gamma}}}
\newcommand{\Sigmaaa} {\mathrm{\mathbf {\Sigma}}}
\newcommand{\HHH} {\mathrm{\mathbf{H}}}
\newcommand{\bbb}{\mathrm{\mathbf{b}}}
\newcommand{\WWW}{\mathrm{\mathbf{G}}}
\newcommand{\WWWW}{\mathrm{\mathbf{W}}}

\newcommand{\A}{\mathrm{\mathbf{A}}}
\newcommand{\aaa}{\mathrm{\mathbf{a}}}
\newcommand{\h}{\mathrm{\mathbf{h}}}

\newcommand{\s}{\mathrm{\mathbf{s}}}
\newcommand{\vvv}{\mathrm{\mathbf{v}}}

\newcommand{\conv}{\mathrm{\mathbf{conv}}}

\newcommand{\bd}{\mathrm{\mathbf{bd}}}
\newcommand{\interior}{\mathrm{\mathbf{int}}}
\newcommand{\ccc}{\mathrm{\mathbf{c}}}
\newcommand{\x}{\mathrm{\mathbf{x}}}
\newcommand{\yyy}{\mathrm{\mathbf{y}}}
\newcommand{\uuu}{\mathrm{\mathbf{u}}}
\newcommand{\C}{\mathbb{C}}
\newcommand{\D}{\mathcal{D}}

\newcommand{\diag}{\mathop{\mathrm{diag}}}

\begin{document}
\maketitle

\IEEEpeerreviewmaketitle

\begin{abstract}
In this paper, we investigate the symbol-level precoding (SLP) design problem  in the downlink of a multiuser multiple-input single-output (MISO) channel. We consider generic constellations with any arbitrary shape and size, and confine ourselves to one of the main categories of constructive interference regions (CIR), namely, distance preserving CIR (DPCIR). We provide a comprehensive study of DPCIRs and derive some properties for these regions. Using these properties, we first show that any signal in a given DPCIR has a norm greater than or equal to the norm of the corresponding constellation point if and only if the convex hull of the constellation contains the origin. It is followed by proving that the power of the noiseless received signal lying on a DPCIR is a monotonic strictly increasing function of two parameters relating to the infinite Voronoi edges. Using the convex description of DPCIRs and their properties, we formulate two design problems, namely, the SLP power minimization with signal-to-interference-plus-noise ratio (SINR) constraints, and the SLP SINR balancing problem under max-min fairness criterion. The SLP power minimization based on DPCIRs can straightforwardly be written as a quadratic program (QP). We provide a simplified reformulation of this problem which is less computationally complex. The SLP max-min SINR, however, is non-convex in its original form, and hence difficult to tackle. We propose several alternative optimization approaches, including semidefinite program (SDP) formulation and block coordinate descent (BCD) optimization. We discuss and evaluate the loss due to the proposed alternative methods through extensive simulation results.

\end{abstract}



\section{Introduction} \label{sec:intro}
Multiuser interference (MUI) is a major performance limiting factor in the downlink of multiuser systems which may adversely affect the achievable transmission rate of individual users. One approach to mitigate the MUI is to precompensate for its undesired effect on the received signal through some signal processing at the transmitter \cite{tb_opt}, which is commonly known as multiuser precoding.
In general, multiuser precoding design can be expressed as a constrained optimization problem \cite{tb_convex, tb_sol_str}. The design problem aims at keeping a balance between some system-centric and user-centric objectives/requirements, depending on the network's operator strategy. Power and sum-rate are often regarded as system-centric criteria \cite{tb_coor}. Transmit power is considered, for example, to control the inter-cell interference in multicell wireless networks, and sum-rate is a measure of the overall system performance. On the other hand, as a user-centric criterion, signal-to-interference-plus-noise ratio (SINR) is an effective measure of quality-of-service (QoS) in multiuser interference channels \cite{tb_mul}. In particular, both bit error rate (BER) and capacity, which are two relevant criteria from a practical point of view, are closely related with maximizing SINR \cite{tb_conic}. Taking into account different types of optimization criteria, some well-known formulations for the multiuser precoding problem are QoS-constrained power minimization \cite{tb_eemax, tb_vis}, SINR balancing \cite{tb_conic, tb_sinr, tb_itr}, and (weighted) sum-rate maximization \cite{tb_coor, wsrmax, tb_has}. In this paper, we primarily focus on the power minimization problem with SINR constraints and the SINR balancing problem based on max-min fairness criterion.

Conventional multiuser precoding techniques try to exploit the knowledge of the channel in order to suppress the MUI. A crucial assumption is therefore the availability of instantaneous or stochastic channel state information (CSI) at the transmitter \cite{mimo_big}. However, the MUI may not always considered to be harmful; on the contrary, following the notion of constructive interference \cite{slp_cdma}, one can turn the MUI into a useful source of signal power instead of treating it as an unwanted distortion \cite{tb_green}. To gain benefit from the potential advantage of constructive interference, it has been recently suggested to design the precoder on a symbol-level basis as a promising alternative to linear block-level precoding \cite{slp_rot, slp_chr, slp_con}. Such a design concept is referred to as symbol-level precoding (SLP). Beside the CSI, the symbol-level design also requires the instantaneous data information (DI) of all users which is readily available at the transmitter. When compared to conventional schemes, it has been shown that significant gains can be achieved at the expense of slightly higher transmitter complexity \cite{slp_rot}, but without re-designing the receiver. While the linear structure of the precoder can be preserved under SLP, one may also form a virtual multicast formulation to directly find the optimal transmit vector, as proposed in \cite{slp_con}, instead of designing the precoder.

The symbol-level design of a multiuser precoder generally involves an optimization problem for each possible combination of the users' symbols. The optimization constraints are so designed to push each user's (noiseless) received signal to a predefined region, called constructive interference region (CIR), enhancing (or guaranteeing a certain level of) detection accuracy. Therefore, the constraints, and hence the SLP problem, highly depend on the adopted constellation. Furthermore, the objective function and the constraints may vary for different problems having particular design criterion and requirements. The SLP problem minimizing the total transmit power has been studied for various constellations, including PSK \cite{slp_chr, slp_con, slp_chr_psk, slp_sec}, QAM \cite{slp_multi, slp_qam}, and APSK \cite{slp_apsk}. For PSK constellations, the minimization of peak per-antenna transmit power is addressed in \cite{slp_pp}. A generic formulation for power minimization problem, not depending on constellation shape and order, is presented in \cite{slp_gen} for both total and peak per-antenna power constraints.

SINR balancing in multiuser multiple-input single-output (MISO) channels is in general more challenging and has been widely investigated for conventional precoding techniques. This problem has been addressed in both multicast (single data stream) and unicast (independent data streams) downlink scenarios \cite{tb_mul, tb_conic, tb_sinr, tb_app}. The problem is not convex in general and is known to be NP-hard \cite{tb_mul}. Several alternate optimization approaches have been proposed in the literature. We kindly refer the readers to \cite{slp_spawc} for a short review on SINR balancing in conventional precoding. In particular, for downlink unicast channels, it is shown in \cite{tb_conic} that the power minimization and the max-min SINR are inverse problems. 

The SINR balancing problem for SLP schemes has not been addressed extensively in the literature. In \cite{slp_con}, the non-convex SLP max-min SINR is solved using its relation to the power minimization via a bisection search. The method is only applicable to PSK constellations (more precisely, to constant envelope modulations) and suffers from high computational complexity. This problem is also addressed in \cite{slp_chr} and a second-order cone program (SOCP) formulation is proposed for PSK constellations. Nevertheless, there is no general solution method or convex formulation for the SLP max-min SINR problem being valid for all generic constellations. 

The main contributions of this paper are as follows:
\begin{itemize}
	\item[1)] We develop the previous work in \cite{slp_gen} through fully characterizing a general family of CIRs, namely, distance preserving CIR (DPCIR). We derive some properties for these regions which apply to any given constellation. The main property states that the norm of any signal in a given unbounded CIR is a monotonically increasing function of two parameters related to the corresponding infinite Voronoi edges, under the necessary and sufficient condition that the convex hull of the constellation contains the origin.
	\item[2)] We study the SLP design criterion from a system-level point of view and discuss the feasibility of QoS provisioning in a resource-constrained multiuser downlink channel through deriving a sufficient feasibility condition. This is followed by providing some reformulations of the DPCIR-based SLP power minimization problem.
	\item[3)] Using the properties of DPCIRs, we show that by fixing a subset of variables in the optimization problem, the SLP max-min SINR can be reduced to a convex problem. Based on this, we bound the search interval to find an approximate solution from a finite discretized candidate set. We further simplify the solution method by providing alternative, but less computationally expensive, optimization approaches in order to achieve sub-optimal solutions for the original problem. Two methods are proposed and evaluated, namely, semidefinite (SDP) formulation and block coordinate descent (BCD) optimization.
	\item[4)] We arrange all the optimization problems in a general form which is indifferent to the type of constellation.
\end{itemize}

The remainder of this paper is organized as follows. In Section \ref{sec:sysmodel}, we describe our system model and define the problems of interest. In Section \ref{sec:cir}, we overview the DPCIRs and characterize them for any given constellation point. We further derive and prove some properties for these regions. We address the SLP design problems in Section \ref{sec:slp}, which includes discussions on the power minimization and proposing alternative solution methods for the SINR balancing. In Section \ref{sec:sim}, we provide some simulation results. Finally, we conclude the paper in Section \ref{sec:conc}.

The following notations are used in the rest of this paper. We use uppercase and lowercase bold-faced letters to denote matrices and vectors respectively, and lowercase normal letters to denote scalars. For matrices and vectors, $[\,\cdot\,]^H$ and $[\,\cdot\,]^T$ denote conjugate transpose and transpose operators, respectively. For vectors, $\|\cdot\|_2$ and $\|\cdot\|_\infty$ represent the $l_2$ norm and the $l_\infty$ norm, and $\succeq$ (or $\succ$) denotes componentwise inequality. For any vector $\vvv$, $\mathrm{diag}(\vvv)$ represents a square matrix with $\vvv$ on its main diagonal and zero off-diagonal elements. For a group of vectors $\vvv_1,...,\vvv_K$, $\mathrm{blkdiag}(\vvv_1,...,\vvv_K)$ represents a square diagonal matrix where its diagonal elements are $\mathrm{diag}(\vvv_1),...,\mathrm{diag}(\vvv_K)$. Operators $|\cdot|$, $\Re\{\cdot\}$, $\Im\{\cdot\}$ and $(\cdot)^*$ denote the respectively amplitude, real part, imaginary part and conjugate of a complex argument. Symbols $\mathbf{0}$, $\mathbf{1}$ and  $\mathbf{I}$ stand for all-zeros vector, all-ones vector and identity matrix of appropriate dimension. For any set $\mathcal{A}$, $|\mathcal{A}|$ denotes the cardinality of $\mathcal{A}$. $\mathbb{R}$ and $\mathbb{C}$ represent the sets of real and complex numbers, and $\mathbb{R}_+$ and $\mathbb{R}_{++}$ represent the sets of non-negative and positive real numbers, respectively. The expectation operator is denoted by $\mathbb{E}\{\cdot\}$.


\section{System Model and Problem Definition} \label{sec:sysmodel}

We consider the downlink of a multiuser MISO unicast channel, where a single base station (BS) sends independent data streams to $K$ users. The BS is equipped with $N$ transmit antennas while each user has a single receive antenna. A block fading channel is assumed between the BS's transmit antennas and the $k$-th user, where the channel vector is denoted by $\h_k\in\C^{1\times N}$. It is further assumed that perfect channel knowledge is available to the BS.

At a given symbol time, $K$ independent symbols are to be sent to $K$ users (throughout the paper, we drop the symbol's time index to simplify the notation). We collect these symbols in users' symbol vector $\s=[s_1,\ldots,s_K]^T\in\C^{K\times1}$ with $s_k$ denoting the symbol intended for the $k$-th user. Each symbol $s_k$ is drawn from a finite equiprobable two-dimensional constellation set. Without loss of generality, we assume an $M$-ary constellation set $\chi=\{x_i|x_i\in\mathbb{C}\}_{i=1}^M$ with unit average power for all $K$ users. 
The user's symbol vector $\s$ is mapped onto $N$ transmit antennas. This is done by a symbol-level precoder, yielding the transmit vector $\uuu=[u_1,\ldots,u_K]^T\in\C^{N\times1}$, as depicted in Fig. \ref{fig:0}. The received signal at the receiver of user $k$ is then
\begin{equation}\label{eq:r}
r_k = \h_k\uuu+w_k, k=1,...,K,
\end{equation}
where $w_k\sim\mathcal{CN}(0,\sigma_k^2)$ is the complex additive white Gaussian noise at the $k$-th receiver. From the received scalar $r_k$, the user $k$ may detect its own symbol $s_k$ by applying the single-user maximum-likelihood (ML) decision rule.

The functionality of a symbol-level precoder is to instantaneously design the signal to be transmitted at each symbol time based on a constrained optimization problem. The solution of this problem, i.e., the transmit vector $\uuu$, is in general a function of instantaneous DI and CSI as well as the set of given system constraints or user-specific requirements.

\begin{figure}
	\centering
	\includegraphics[width=0.7\columnwidth]{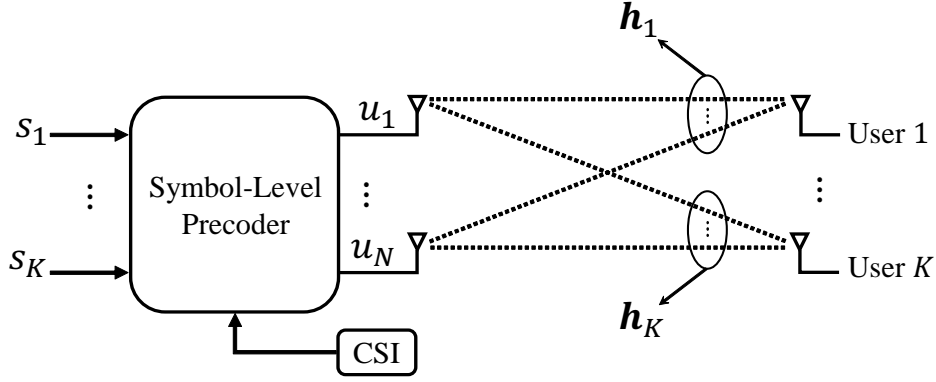}
	\caption{SLP-based diagram of a downlink multiuser MISO unicast channel.}
	\label{fig:0}
\end{figure}
 
In the SLP power minimization problem, the user-specific requirements are individual SINR thresholds that guarantee the reliable communication for each user. These thresholds impose some constraints on the design problem. To be more specific, the SINR-related constraints can be expressed as
\begin{equation}\label{eq:raw}
\uuu^H\h_k^H\h_k\uuu \geq \sigma_k^2\gamma_k\, s_k^* s_k,\;k = 1,...,K,
\end{equation}
where $\gamma_k$ is the SINR threshold for the $k$-th user. It should be noted that the SINR thresholds $\{\gamma_k\}_{k=1}^K$ typically refer to the long-term (e.g., frame-level) SINRs, i.e., the average received SINR over all the symbols in a frame. However, for sufficiently large frames (which is often the case in practice) we have $\mathbb{E}\{s_k^* s_k\}\rightarrow1, k = 1,...,K$, and hence the symbol-level constraints \eqref{eq:raw} satisfy the frame-level SINR thresholds $\mathbb{E}\{\uuu^H\h_k^H\h_k\uuu\} \geq \sigma_k^2\gamma_k, k = 1,...,K$, with the expectations being taken over the entire frame. Therefore, one can think of the symbol-level constraints \eqref{eq:raw} as a conservative way to meet the frame-level SINR requirements.

Considering \eqref{eq:raw}, the SLP power minimization problem for a generic constellation can be formulated as
\begin{equation}\label{eq:raw2}
\begin{aligned}
\underset{\uuu}{\mathrm{minimize}} & \quad f(\uuu) \\
\mathrm{s.t.} &\quad \h_k\uuu\in \sigma_k\sqrt{\gamma_k}\;\D_k,\;k = 1,...,K, \\
\end{aligned}
\end{equation}
where $\D_k$ represents the CIR associated with symbol $s_k$, and the objective function $f(\uuu)$ can be either $\uuu^H\uuu$ or $\|\uuu\|_\infty^2$ depending on whether the total or the peak (per-antenna) transmit power is minimized. A sufficient condition under which any solution to \eqref{eq:raw2} satisfies the SINR constraints \eqref{eq:raw} is that the amplitude of any point in $\D_k$ is at least equal to $|s_k|=\sqrt{s_k^* s_k}$, for all $k=1,...,K$.

The SLP SINR balancing problem, on the other hand, aims to service all the users in a fair manner while a system-centric restriction is usually considered to be the total transmit power. In particular, under the max-min fairness criterion, the goal is to maximize the worst SINR among all users subject to the power constraint. This leads to the following general design formulation
\begin{equation}\label{eq:raw3}
\begin{aligned}
\underset{\uuu}{\mathrm{maximize}} & \quad \underset{k}{\min}{\left\{\frac{\uuu^H\h_k^H\h_k\uuu}{\sigma_k^2}\right\}_{k=1}^K} \\
\mathrm{s.t.} &\quad \h_k\uuu\in \sigma_k\;\D_k,\;k = 1,...,K, \\
& \quad \uuu^H\uuu \leq P_{\mathrm{max}},
\end{aligned}
\end{equation}
where $P_{\mathrm{max}}$ is the downlink total power budget.

We will formulate and discuss both the problems \eqref{eq:raw2} and \eqref{eq:raw3} in Section \ref{sec:sim}, assuming the CIRs to be distance preserving. To this end, we first present a detailed study of the CIRs which enables us to exploit their properties in order to properly form the constraints of the SLP optimization problem.
\section{Distance Preserving Constructive Interference Regions}\label{sec:cir}

In this section, we provide an in-depth overview of a general category of CIRs, namely, DPCIR, and develop their characterization in \cite{slp_gen} by deriving some of their properties. The main results of this section are stated in Lemma \ref{lem:2}, Lemma \ref{lem:3} and Theorem \ref{thm:1}. The proofs have been previously presented in \cite{slp_spawc}. For the sake of completeness, we provide the proofs also in this paper in appendices \ref{app:lem2}-\ref{app:thm}.

Hereafter, we denote each complex-valued constellation point by its equivalent real-valued vector notation, hence the set of points in $\chi$ is denoted by $\{\x_i|\x_i\in\mathbb{R}^2\}_{i=1}^M$.
For the equiprobable constellation set $\chi$, the ML decision rule corresponds to the Voronoi regions of $\chi$ which are bounded by hyperplanes. For a given constellation point $\x_i$ and one of its neighboring points $\x_j$, the hyperplane separating the Voronoi regions of $\x_i$ and $\x_j$ is given by $\{\x\mid \x\in\mathbb{R}^2, \aaa_{i,j}^T \x=b_{i,j}\}$, where $\aaa_{i,j}=\x_i-\x_j$ (or any non-zero scalar multiplication of $\x_i-\x_j$), and $b_{i,j}=\aaa_{i,j}^T(\x_i+\x_j)/2$. This hyperplane indicates a decision boundary (Voronoi edge) between $\x_i$ and $\x_j$, which splits $\mathbb{R}^2$ plane into two halfspaces. The closed halfspace that contains the decision region of $\x_i$ is represented as
\begin{equation}\label{eq:hs}
\mathcal{H}_{i,j,\text{\tiny ML}}=\{\x\mid \x\in\mathbb{R}^2, \aaa_{i,j}^T \x\geq b_{i,j}\},
\end{equation}
where $\aaa_{i,j}$ is the inward normal and $b_{i,j}$ determines the offset from the origin. The Voronoi region of $\x_i$ is then given by intersecting all the halfspaces \eqref{eq:hs} resulting from the neighboring points of $\x_i$, i.e.,
\begin{equation}\label{eq:cap}
\begin{aligned}
\mathcal{D}_{i,\text{\tiny ML}}&=\bigcap_{\x_j\in\mathcal{S}_i}\mathcal{H}_{i,j,\text{\tiny ML}}\\
&=\left\{\x\mid \x\in\mathbb{R}^2, \aaa_{i,j}^T \x\geq b_{i,j}, \forall j\in\mathcal{J}_i\right\},
\end{aligned}
\end{equation}
where $\mathcal{J}_i=\{j|\x_j\in\mathcal{S}_i\}$ and $\mathcal{S}_i$ denotes the set of neighboring points of $\x_i$ with $|\mathcal{S}_i|=|\mathcal{J}_i|=M_i$ . Each Voronoi region can be either an unbounded or bounded polyhedron, depending on the relative location of $\x_i$ in $\chi$. It can be easily verified that the Voronoi regions are always convex sets \cite{convex_boyd}. The Voronoi region \eqref{eq:cap} can be expressed in a more compact form as
\begin{equation}\label{eq:capc}
\mathcal{D}_{i,\text{\tiny ML}}=\left\{\x\mid \x\in\mathbb{R}^2, \A_i \x\succeq \bbb_i\right\},
\end{equation}
where $\A_i\in\mathbb{R}^{M_i\times2}$ and $\bbb_i\in\mathbb{R}^{M_i}$ contain $\aaa_{i,j}^T$ and $b_{i,j}$, respectively, for all $j\in\mathcal{J}_i$.
The halfspace representation of ML decision regions in \eqref{eq:capc} can be used to describe the DPCIRs \cite{slp_gen}, as will be explained in the following.

The distance preserving margin between $\x_i$ and $\x_j$, by definition, is equal to $\frac{d_{i,j}}{2}$, where $d_{i,j}$ denotes the original distance between the two constellation points. Accordingly, given the Voronoi hyperplane $\{\x\mid \x\in\mathbb{R}^2, \aaa_{i,j}^T \x=b_{i,j}\}$, the corresponding distance preserving hyperplane can be represented by $\left\{\x\mid \x\in\mathbb{R}^2, \aaa_{i,j}^T \x=b_{i,j}+c_{i,j}\right\}$, where $c_{i,j}=\frac{d_{i,j}\|\aaa_{i,j}\|_2}{2}$. These two hyperplanes are parallel to each other with an orthogonal distance of $\frac{c_{i,j}}{\|\aaa_{i,j}\|_2}$ in the direction of $\aaa_{i,j}$. The resulting closed halfspace is then given by
\begin{equation}\label{eq:capdp}
\mathcal{H}_{i,j,\text{\tiny DP}}=\left\{\x\mid \x\in\mathbb{R}^2, \aaa_{i,j}^T \x\geq b_{i,j}+c_{i,j}\right\}.
\end{equation}
Intersecting \eqref{eq:capdp} over all the neighboring points of $\x_i$ gives the associated DPCIR as
\begin{equation}\label{eq:dpcirlinineq}
\begin{aligned}
\mathcal{D}_{i,\text{\tiny DP}}&=\bigcap_{j\in\mathcal{J}_i}\mathcal{H}_{i,j,\text{\tiny DP}}\\
&=\left\{\x\mid \x\in\mathbb{R}^2, \aaa_{i,j}^T \x\geq b_{i,j}+c_{i,j}, \forall j\in\mathcal{J}_i\right\},\\
&=\left\{\x\mid\x\in\mathbb{R}^2, \A_i \x\succeq \bbb_i+\ccc_{i}\right\},
\end{aligned}
\end{equation}
where $\ccc_{i}\in\mathbb{R}^{M_i}_+$ is the vector containing $\frac{d_{i,j}\|\aaa_{i,j}\|_2}{2}$ for all $j\in\mathcal{J}_i$. Similar to $\mathcal{D}_{i,\text{\tiny ML}}$, the region $\mathcal{D}_{i,\text{\tiny DP}}$ is given by the intersection of a number of closed halfspaces and thus is a polyhedron. Furthermore, the bounding hyperplanes of $\mathcal{D}_{i,\text{\tiny DP}}$ are parallel to their corresponding Voronoi edges, i.e., they have the same inward normals $\aaa_{i,j}, j\in\mathcal{J}_i$. It is straightforward to show that the following properties hold for DPCIRs:
\begin{property}\label{pro:0}
For all $\x_i\in\chi$ and any $\x\in\mathcal{D}_{i,\text{\tiny\normalfont DP}}$, we have
\begin{itemize}
\item[i.]$\mathcal{D}_{i,\textrm{\tiny\normalfont DP}}\subseteq\mathcal{D}_{i,\text{\tiny\normalfont ML}}$.
\item[ii.]$\|\x-\yyy\|_2\geq\|\x_i-\x_j\|_2=d_{i,j}$, $\forall \x_j\in\chi$ and any $\yyy\in\mathcal{D}_{j,\text{\tiny\normalfont DP}}$.
\end{itemize}
\end{property}
\noindent A special case of Property \ref{pro:0}-ii for $\yyy=\x_j$ becomes
\begin{equation}\label{eq:pro0cons}
\|\x-\x_j\|_2\geq \|\x_i-\x_j\|_2, \forall \x_j\in\chi,
\end{equation}
where \eqref{eq:pro0cons} holds with equality only when $\x=\x_i$.

The convex hull $\conv\chi$, i.e., the smallest convex set containing $\chi$, can be easily derived from the constellation set $\chi$. The set of points belonging to the boundary of $\conv\chi$ is denoted by $\bd\chi$, and the set of interior points of $\conv\chi$, i.e., $\conv\chi\backslash\bd\chi$, is denoted by $\interior\chi$. A typical illustration of these sets for the optimized 8-ary constellation \cite{cons_opt} is shown in Fig. \ref{fig:1} (a). It follows from \eqref{eq:dpcirlinineq} that if $\mathcal{D}_{i,\text{\tiny ML}}$ is bounded, then $\mathcal{D}_{i,\text{\tiny DP}}=\x_i$, which means that all the inequalities are satisfied with equality. On the other hand, for an unbounded $\mathcal{D}_{i,\text{\tiny ML}}$, the associated $\mathcal{D}_{i,\text{\tiny DP}}$ is an unbounded polyhedron (more specifically, a polyhedral angle as depicted in Fig. \ref{fig:1} (a)) which is uniquely characterized using the two following lemmas.

\begin{figure}
	\begin{subfigmatrix}{3}
		\subfigure[]{\includegraphics[width=.3\columnwidth]{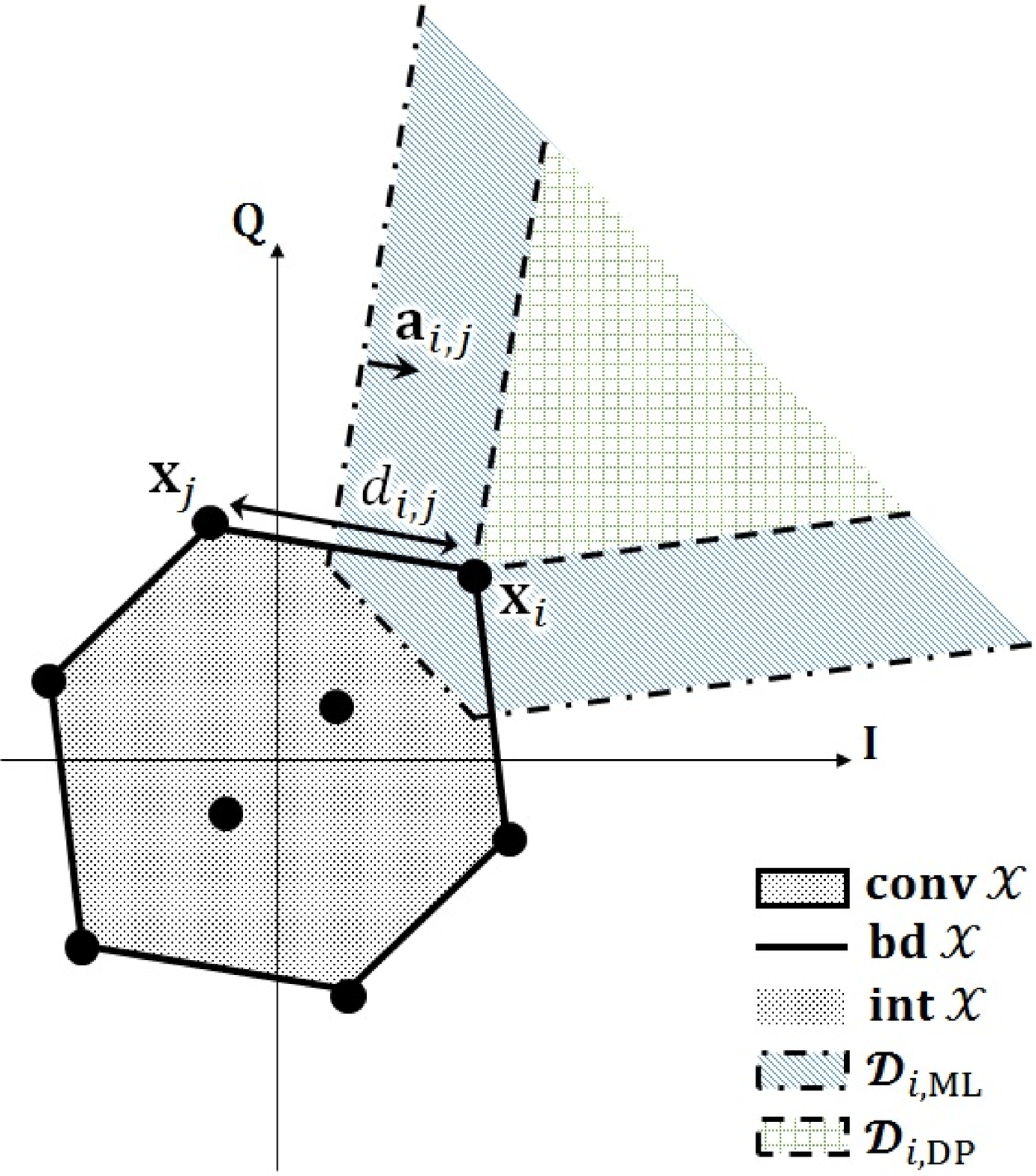}}
		\subfigure[]{\includegraphics[width=.3\columnwidth]{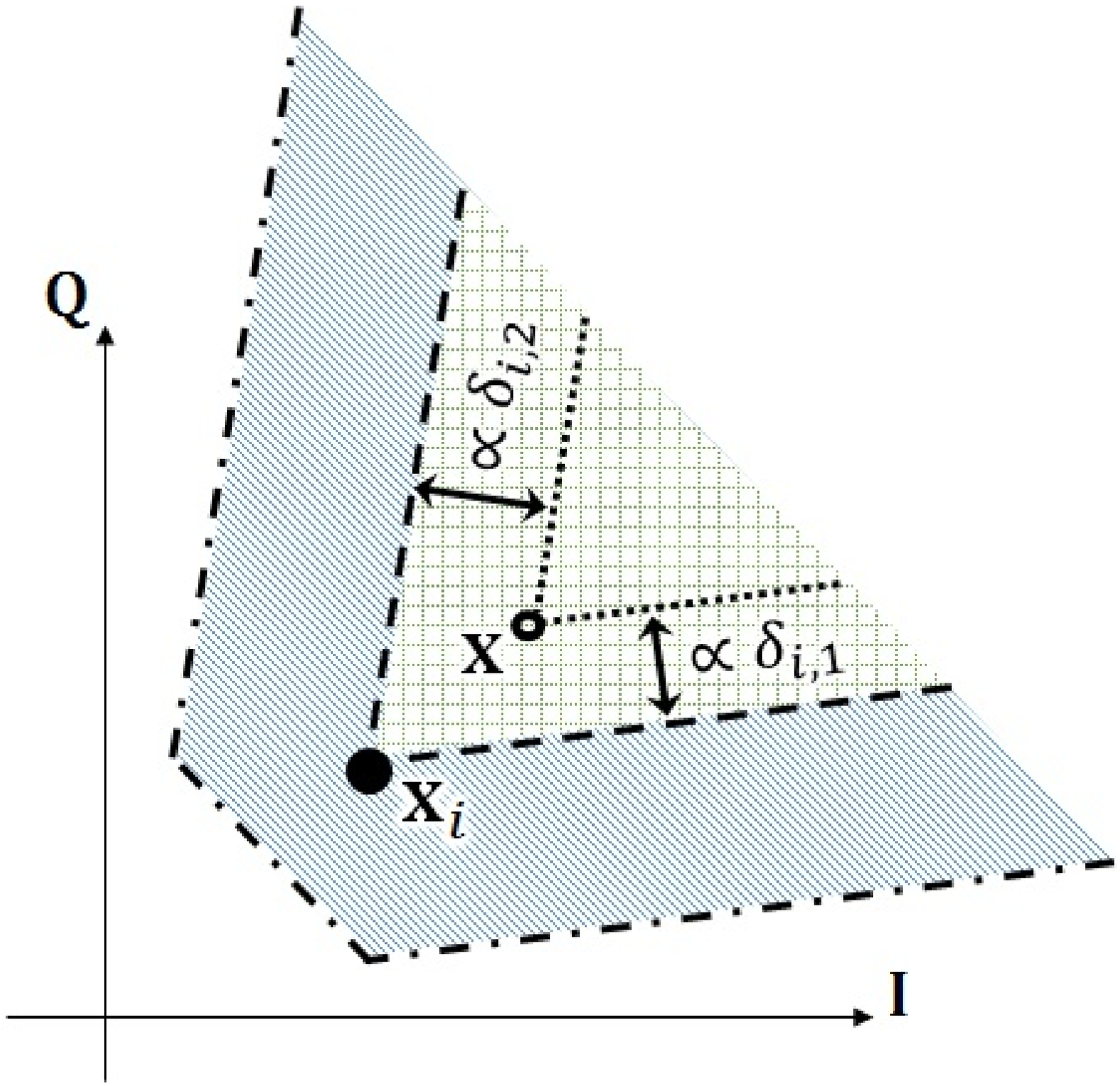}}
	\end{subfigmatrix}
	\caption{The optimized 8-ary constellation. (a) A boundary points $\x_i$ with unbounded Voronoi region; $\mathcal{D}_{i,\text{\tiny DP}}$ is a polyhedral angle with two infinite edges starting from $\x_i$. (b) Any point $\x\in\mathcal{D}_{i,\text{\tiny DP}}$ can be represented by \eqref{eq:del} if one displaces the two infinite bounding hyperplanes, each of which by an orthogonal distance proportional to $\delta_{i,1}$ or $\delta_{i,2}$.}
	\label{fig:1}
\end{figure}

\begin{lemma}\label{lem:1}
A point $\x_i\in\chi$ lies on the boundary of (or is a vertex of) $\conv\chi$ iff its Voronoi region $\mathcal{D}_{i,\text{\tiny\normalfont ML}}$ is unbounded \cite[Lemma 2.2]{vor_diag}.
\end{lemma}
\begin{lemma}\label{lem:2}
For every $\x_i\in\chi$ with unbounded $\mathcal{D}_{i,\text{\tiny ML}}$, $\mathcal{D}_{i,\text{\tiny DP}}$ is a polyhedral angle with a vertex at $\x_i$, and each of its edges is perpendicular to one of the two line segments connecting $\x_i$ to its two neighboring points on $\bd\chi$.
\end{lemma}

\begin{proof}
See Appendix \ref{app:lem2}.
\end{proof}

Lemma \ref{lem:2} implicitly states that neither changing the location of any constellation point $\x_j\in\interior\chi$ nor adding a new constellation point on $\bd\chi$ does not affect $\mathcal{D}_{i,\text{\tiny DP}}$ for any $\x_i\in\bd\chi$, as they both keep the direction of $\aaa_{i,j}$ unchanged for all $\x_j\in\mathcal{S}_i\cap\bd\chi$. 

Next, we prove that the norm of any point in a DPCIR is always greater than or equal to the norm of the corresponding vertex if and only if the convex hull of the constellation includes the origin. It should be noted that this is a rather light condition, as all well-known constellations in the literature with $M\geq4$ have at least one point in each quadrant and therefore their convex hull contains the origin. 

\begin{lemma}\label{lem:3}
For any constellation point $\x_i\in\chi$, we have $\|\x\|\geq\|\x_i\|,\forall\x\in\mathcal{D}_{i,\text{\tiny\normalfont DP}}$ iff $\conv\chi$ contains the origin. Equality is achieved only when $\x=\x_i$.
\end{lemma}

\begin{proof}
	See Appendix \ref{app:lem3}.
\end{proof}

To proceed, it is more convenient to express the linear inequalities of \eqref{eq:dpcirlinineq} by an equivalent set of linear equations as
\begin{equation}\label{eq:dpcirlineq}
\mathcal{D}_{i,\text{\tiny DP}}=\Big\{\x\mid\x\in\mathbb{R}^2, \A_i \x=\bbb_i+\ccc_{i}+\Deee_i, \Deee_i\in\mathbb{R}^{M_i}_+ \Big\}.
\end{equation}
The linear equations in \eqref{eq:dpcirlineq} indicate that any $\x\in\mathcal{D}_{i,\text{\tiny DP}}$ can be realized as the intersection point of $M_i$ hyperplanes, each of which is parallel to a boundary hyperplane of $\mathcal{D}_{i,\text{\tiny DP}}$ but has a different offset due to the term $\Deee_i$.\\

\noindent {\bf{Remark 1.}} It is easy to verify that a hyperplane in a set of hyperplanes describing the boundaries of a polyhedron is redundant if the corresponding polyhedron remains unchanged by removing the hyperplane \cite[p. 9]{comb_opt}. Therefore, we can remove from \eqref{eq:dpcirlineq} the equalities that come from a redundant hyperplane. Using this and based on Lemma \ref{lem:2}, for any $\x_i\in\bd\chi$, the associated region $\mathcal{D}_{i,\text{\tiny DP}}$ is spanned by at most two non-negative parameters $\delta_{i,1}$ and $\delta_{i,2}$ corresponding to the two infinite Voronoi edges. Therefore, any point $\x \in \mathcal{D}_{i,\text{\tiny DP}}$ can be specified by
\begin{equation}\label{eq:del}
\Deee_i=[\delta_{i,1},\delta_{i,2}]^T\in\mathbb{R}^2_+, \forall \x_i\in\bd\chi,
\end{equation}
which makes $\A_i$ a $2\times2$ full-rank, and hence invertible, matrix. 

It should be pointed out that this representation covers the special case with the two infinite Voronoi edges being parallel to each other (e.g., QAM constellations). In such case, both $\delta_{i,1}$ and $\delta_{i,2}$ are constrained to be always zero; but the region $\mathcal{D}_{i,\text{\tiny DP}}$, which is a half-line starting from the constellation point $\x_i$, can be spanned by a non-negative parameter indicating the displacement of a virtual hyperplane orthogonal to the two existing infinite Voronoi edges. Thereby, any point $\x\in\mathcal{D}_{i,\text{\tiny DP}}$ is represented by, for example,  $\Deee_i=[\delta_{i,1},0]^T\in\mathbb{R}^2_+$, which preserves the invertibility of $\A_i$.

It is important to notice that our definitions are presented for two-dimensional constellations. For pulse amplitude modulation (PAM) schemes, where the constellation is one-dimensional, one may define the same concept by embedding PAM in $\mathbb{R}^2$. 
However, in this case \eqref{eq:del} fails to span the entire region $\mathcal{D}_{i,\text{\tiny DP}}$ since a single hyperplane solely contributes to $\mathcal{D}_{i,\text{\tiny DP}}$ for each $\x_i\in\chi$, as illustrated in Fig. \ref{fig:3} for 4-PAM. Any point $\x$ belonging to the region $\mathcal{D}_{i,\text{\tiny DP}}$ associated with either of the two outer symbols can be described as, for example, $\Deee_i=[\delta_{i,1},\delta_{i,2}]^T$, where $\delta_{i,1}\in\mathbb{R}_+$ corresponds to the single hyperplane and $\delta_{i,2}\in\mathbb{R}$ displaces a virtual hyperplane granting the second basis to span $\mathcal{D}_{i,\text{\tiny DP}}$. For those constellation points other than the two outer symbols, any $\x\in\mathcal{D}_{i,\text{\tiny DP}}$ is represented by $\Deee_i=[0,\delta_{i,2}]^T$, where $\delta_{i,2}\in\mathbb{R}$ grants the only basis spanning $\mathcal{D}_{i,\text{\tiny DP}}$ (see Fig. \ref{fig:3}). 

\begin{figure}
	\centering
	\includegraphics[width=0.6\columnwidth]{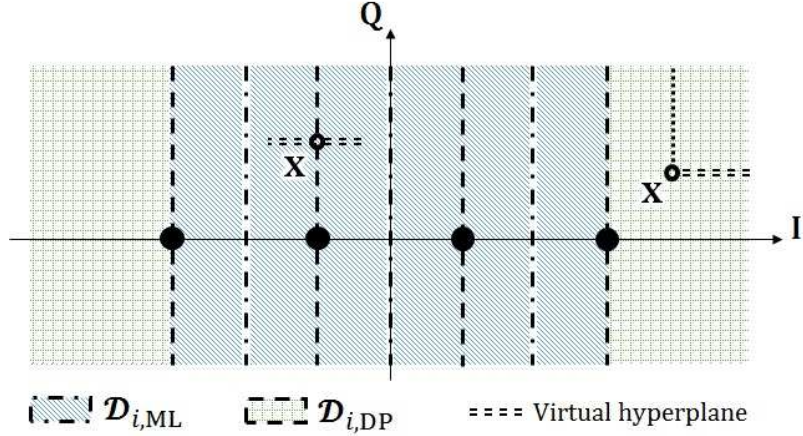}
	\caption{4-PAM constellation in $\mathbb{R}^2$; a virtual hyperplane is needed to specify any point $\x\in\mathcal{D}_{i,\text{\tiny DP}}$.}
	\label{fig:3}
\end{figure}

Finally, we state the following theorem which will help us to formulate the max-min SINR and power minimization problems for SLP in the next section.
\begin{theorem}\label{thm:1}
For any constellation point $\x_i\in\bd\chi$ with $\mathcal{D}_{i,\text{\tiny\normalfont DP}}$ as expressed in \eqref{eq:dpcirlineq}, function $f(\x)=\|\x\|$ over its domain $\mathcal{D}_{i,\text{\tiny\normalfont DP}}$ is a monotonic strictly increasing function of each element of $\Deee_i$ iff $\conv\chi$ contains the origin.
\end{theorem}
\begin{proof}
See Appendix \ref{app:thm}
\end{proof}

It is worthy to mention that assuming the CIRs to be either union bound (UBCIR) or minimum distance preserving (MDPCIR), as defined in \cite{slp_gen}, a variant of Theorem \ref{thm:1} still holds. In both cases, the norm of any point belonging to these regions is strictly increasing in exactly two elements related to the two infinite Voronoi edges.

\section{Symbol-level Precoding Design Problem}\label{sec:slp}

In this section, by using the properties of DPCIRs proved in the previous section, we formulate the optimization problem of multiuser precoding on a symbol-level basis. In particular, we are interested in two well-known design problems, namely, power minimization and SINR balancing. As discussed in Section \ref{sec:cir}, the DPCIRs can readily be obtained for all generic constellations since they depend only on the Voronoi regions. This enables us to arrange the optimization problems in a general form which is indifferent to the type of constellation. 

Throughout this section, for any user $k=1,...,K$, the symbol $s_k$ corresponds to one of the points $\{\x_i\}_{i=1}^M$ in $\chi$. We denote by $i_k$ the index of the constellation point corresponding to $s_k$, i.e.,
$$\begin{bmatrix}\Re\{s_k\}\\ \Im\{s_k\}\end{bmatrix}=\x_{i_k}, i_k\in\{1,...,M\}.$$
Furthermore, we define the index set $\mathcal{K}\!=\!\{k|\x_{i_k}\!\in\!\bd\chi\}$ referring to those users with a symbol in the boundary of constellation $\chi$. In the following, we rearrange vectors $\uuu$ and $\h_k$ as
\begin{equation}
\nonumber
\tilde{\uuu}=
\begin{bmatrix}
\Re\{\uuu\}\\
\Im\{\uuu\}
\end{bmatrix}\in\mathbb{R}^{2N\times1},
\end{equation}
\begin{equation}
\nonumber
\HHH_k=
\begin{bmatrix}
\Re\{\h_k\} & -\Im\{\h_k\}\\
\Im\{\h_k\} & \Re\{\h_k\}
\end{bmatrix}\in\mathbb{R}^{2\times2N}, k=1,...,K,
\end{equation}
respectively, such that $\HHH_k\tilde{\uuu}$ represents the noise-free received signal at the $k$-th user's receiver. It is easy to check that $\uuu^H\uuu=\tilde{\uuu}^T\tilde{\uuu}$. We further denote by
\begin{equation}
\nonumber
\WWW=
\begin{bmatrix}
\A_{i_1}\HHH_1\\
\vdots \\
\A_{i_K}\HHH_K
\end{bmatrix}\in\mathbb{R}^{L\times2N}, \bbb=[\bbb_{i_1},...,\bbb_{i_K}]^T\in\mathbb{R}^L,
\end{equation}
\begin{equation}
\nonumber
\ccc=[\ccc_{i_1},...,\ccc_{i_K}]^T\in\mathbb{R}^L, \Deee=[\Deee_{i_1},...,\Deee_{i_K}]^T\in\mathbb{R}^L,
\end{equation}
the vectors and matrices collecting the CIR parameters for all $K$ users, where $L=\sum_{k=1}^K M_{i_k}$ in the general case, but is reducible to $L=2|\mathcal{K}|+\sum_{k\notin\mathcal{K}} M_{i_k}$ due to Remark 1.

\subsection{DPCIR-based SLP Power Minimization}

First, we consider a power-restricted scenario in which the downlink transmission should provide each user with its minimum required SINR, while the BS is subject to a total power constraint $P$. For the rationale behind the power minimization problem we kindly refer the reader to\cite{tb_sinr}. In such cases, designing the DPCIR-based precoder involves solving the following feasibility problem
\begin{equation}\label{eq:feas}
\begin{aligned}
\mathrm{find} & \quad \tilde{\uuu}\\
\mathrm{s.t.} & \quad \A_{i_k}\HHH_k\tilde{\uuu}=\sigma_k\sqrt{\gamma_k}\;(\bbb_{i_k}+\ccc_{i_k})+\Deee_{i_k}, k = 1,...,K,\\
& \quad \Deee_{i_k}\succeq\mathbf{0}, k = 1,...,K,\\
& \quad \tilde{\uuu}^T\tilde{\uuu}\leq P,
\end{aligned}
\end{equation}
in order to see whether the given set of SINR thresholds $\{\gamma_k\}_{k=1}^K$ is achievable, i.e., whether the spatial multiplexing to serve multiple users is meaningful. Otherwise, if there is no solution to \eqref{eq:feas}, the system operator decides to relax the other constraints (e.g., decreasing the number of users $K$, or increasing the power budget). Defining $\Sigmaaa=\mathrm{blkdiag}(\sigma_1,...,\sigma_k)\in\mathbb{R}^{L\times L}$ and $\Gammaaa=\mathrm{blkdiag}(\gamma_1,...,\gamma_k)\in\mathbb{R}^{L\times L}$, the feasibility problem \eqref{eq:feas} can be written, in a compact form, as
\begin{equation}\label{eq:feascomp}
\begin{aligned}
\mathrm{find} & \quad \tilde{\uuu}\\
\mathrm{s.t.} & \quad \WWW\tilde{\uuu} = \Sigmaaa\Gammaaa^{1/2}(\bbb+\ccc)+\Deee,\\
& \quad \Deee\succeq\mathbf{0},\\
& \quad \tilde{\uuu}^T\tilde{\uuu}\leq P.
\end{aligned}
\end{equation}
A sufficient condition under which there exists a feasible point for \eqref{eq:feascomp} can be obtained according to the following proposition.
\begin{proposition}\label{prop:1}
	The feasibility problem \eqref{eq:feascomp} is solvable for $L\leq 2N$ if
	\begin{equation}\label{prop1eq}
	\|\WWW^\dagger\Sigmaaa\Gammaaa^{1/2}(\bbb+\ccc)\|_2^2\leq P,
	\end{equation}
	where $\WWW^\dagger=(\WWW^T\WWW)^{-1}\WWW^T$ is the Moore-Penrose (left) pseudoinverse of $\WWW$.
\end{proposition}
\begin{proof}
	Let $\Deee=\mathbf{0}$, then \eqref{eq:feascomp} reduces to
	\begin{equation}\label{eq:feascompzf}
	\begin{aligned}
	\mathrm{find} & \quad \tilde{\uuu}\\
	\mathrm{s.t.} & \quad \WWW\tilde{\uuu} = \Sigmaaa\Gammaaa^{1/2}(\bbb+\ccc),\\
	& \quad \tilde{\uuu}^T\tilde{\uuu}\leq P.
	\end{aligned}
	\end{equation}
	Now suppose that $\tilde{\uuu}_\mathrm{o} = \WWW^\dagger\Sigmaaa\Gammaaa^{1/2}(\bbb+\ccc)$ is a solution (not necessarily unique) to the system of linear equations
	\begin{equation}\label{eq:feassle}
	\WWW\tilde{\uuu} = \Sigmaaa\Gammaaa^{1/2}(\bbb+\ccc).
	\end{equation}
	 In fact, $\tilde{\uuu}_\mathrm{o}$ coincides with the solution of the well-known zero-forcing (ZF) beamformer \cite{tb_zf} when all the users are allocated identical SINR thresholds. We argue the existence of $\tilde{\uuu}_\mathrm{o}$ as follows. In case $L=2N$, due to the random channel matrices $\HHH_k, k=1,...,K$, matrix $\WWW$ is full-rank almost surely. This means that the probability of \eqref{eq:feassle} having more than one solution is zero. On the other hand, for $L<2N$, \eqref{eq:feassle} expresses an underdetermined system of linear equations for which $\tilde{\uuu}_\mathrm{o}$ is the least-norm solution. Having $\tilde{\uuu}_\mathrm{o}$ as a solution to \eqref{eq:feassle}, if $\tilde{\uuu}_\mathrm{o}^T\tilde{\uuu}_\mathrm{o}\leq P$, then $\tilde{\uuu}_\mathrm{o}$ is a feasible point for \eqref{eq:feascompzf}; this further ensures the feasibility of \eqref{eq:feascomp} since it is a relaxed version of \eqref{eq:feascompzf}. Therefore $$\tilde{\uuu}_\mathrm{o}^T\tilde{\uuu}_\mathrm{o}=(\bbb+\ccc)^T\Gammaaa^{1/2}\Sigmaaa(\WWW\WWW^T)^\dagger\Sigmaaa\Gammaaa^{1/2}(\bbb+\ccc)\leq P$$ is a sufficient condition for the feasibility problem \eqref{eq:feascomp} to be solvable.
\end{proof} 

If a solution to \eqref{eq:feascomp} exists, then the relevant problem is to further reduce the transmit power, which is known as power minimization. The precoder is designed to minimize either the total or the peak (per-antenna) transmit power. The latter objective is more realistic as, in practice, many systems are subject to individual per-antenna power constraints \cite{tb_zf, slp_pp}. Accordingly, the DPCIR-based SLP problem minimizing the total transmit power can be formulated as a standard quadratic program (QP), i.e.,
\begin{equation}\label{eq:pm}
\begin{aligned}
\underset{\tilde{\uuu}, \Deee\succeq\mathbf{0}}{\mathrm{minimize}} & \quad \tilde{\uuu}^T\tilde{\uuu}\\
\mathrm{s.t.} & \quad \WWW\tilde{\uuu} = \Sigmaaa\Gammaaa^{1/2}(\bbb+\ccc)+\Deee.
\end{aligned}
\end{equation}
Denoting the optimal solution of \eqref{eq:pm} by $\tilde{\uuu}^*$, it is naturally expected that $\tilde{\uuu}^{*^T}\tilde{\uuu}^*\leq P$. Replacing $\tilde{\uuu}^T\tilde{\uuu}$ by $\|\tilde{\uuu}\|_{\infty,\mathbb{C}}^2$, the SLP design objective turns to minimize the peak per-antenna transmit power, where by $\|\cdot\|_{\infty,\mathbb{C}}$ we mean the infinity norm over equivalent complex-valued elements. All these variants of the SLP power optimization problem have convex objective functions and constraints, hence are convex, and can efficiently be solved using off-the-shelf methods \cite{convex_boyd}. The feasibility problem \eqref{eq:feascomp} can also be extended to the case with peak per-antenna power constraint if one substitutes $\|\tilde{\uuu}\|_{\infty,\mathbb{C}}^2$ for $\tilde{\uuu}^T\tilde{\uuu}$, and $P/N$ for $P$. Then, the feasibility condition can be written as $\|\WWW^\dagger\Sigmaaa\Gammaaa^{1/2}(\bbb+\ccc)\|_{\infty,\mathbb{C}}^2\leq P/N$. It is worth noting that if this condition holds, then the feasibility condition in Proposition \ref{prop:1} is also satisfied given the norm inequality
$$\|\WWW^\dagger\Sigmaaa\Gammaaa^{1/2}(\bbb+\ccc)\|_2\leq\sqrt{N}\|\WWW^\dagger\Sigmaaa\Gammaaa^{1/2}(\bbb+\ccc)\|_{\infty,\mathbb{C}}.$$

It is possible to further simplify (in terms of computational complexity) the SLP power minimization problem by reducing the number of optimization variables and constraints as below.

\begin{lemma}\label{lem:4}
The QP \eqref{eq:pm} can be reduced to
\begin{equation}\label{eq:pms}
\begin{aligned}
\underset{\Deee\succeq\mathbf{0}}{\mathrm{minimize}} & \quad \left\|\WWW^\dagger\left(\Sigmaaa\Gammaaa^{1/2}(\bbb+\ccc)+\WWWW\Deee\right)\right\|_2^2,\\
\end{aligned}
\end{equation}
for $K\leq N$, where $\WWWW$ is an $L\times L$ diagonal matrix with a diagonal element being one if it corresponds to a user in $\mathcal{K}$, otherwise zero. The optimal transmit vector $\tilde{\uuu}^*$ is then given by
\begin{equation}\label{eq:pmopt}
\tilde{\uuu}^*=\WWW^\dagger\left(\Sigmaaa\Gammaaa^{1/2}(\bbb+\ccc)+\WWWW\Deee^*\right),
\end{equation}
where $\Deee^*$ is the optimum of \eqref{eq:pms}.
\end{lemma}
\begin{proof}
	To verify the equivalency of problems \eqref{eq:pm} and \eqref{eq:pms}, let consider two cases. If $L=2N$, then $\WWW$ is full-rank almost surely, and hence $\WWW^\dagger=\WWW^{-1}$. As a result, the constraint $\WWW\tilde{\uuu} = \Sigmaaa\Gammaaa^{1/2}(\bbb+\ccc)+\Deee$ in \eqref{eq:pm} gives a unique solution for any fixed $\Deee$. In such case, there would be a bijection from $\Deee$ to $\tilde{\uuu}$, which demonstrates that solely optimizing $\Deee$ is equivalent to optimizing both $\tilde{\uuu}$ and $\Deee$. In the other case with $L<2N$, the constraint $\WWW\tilde{\uuu}=\Sigmaaa\Gammaaa^{1/2}(\bbb+\ccc)+\Deee$ may have more than one solution, but \eqref{eq:pmopt} is its least-norm solution which is in line with the objection function of the QP \eqref{eq:pm}. Furthermore, the diagonal matrix $\WWWW$ imposes $\Deee_{i_k}=\mathbf{0}$ for any $k\notin\mathcal{K}$, if exists.
\end{proof}

\subsection{DPCIR-based SLP SINR Balancing}\label{subsec:dpcirsb}

In a downlink scenario where power is a strict transmit restriction, fairness might be a relevant design criterion \cite{tb_conic}. In this paper, we are interested in max-min fairness criterion under which the SLP design problem aims at maximizing the worst SINR among all users, limited by a total transmit power $P$. Assuming the CIRs to be either distance preserving or union bound \cite{slp_gen}, the problem is not convex in its original form. In this section, we first provide an overview and discuss the methods presented in the literature to solve the SLP max-min SINR. Then we derive several alternate convex formulations for this problem. All the proposed methods are simulated in Section \ref{sec:sim} with a detailed discussion on complexity and performance of each solution method. 

One may tackle the SLP max-min SINR by exploiting its connection to the power minimization, as proposed in \cite{slp_con}. By considering the DPCIR-based design as a generalization of \cite{slp_con}, this method iteratively solves

\begin{equation}\label{eq:pmw}
\begin{aligned}
\tilde{\uuu}_\mathrm{PM}(\Gammaaa^*)=
\mathrm{arg}\underset{\tilde{\uuu},\Deee\succeq\mathbf{0}}{\mathrm{min}} & \enspace \tilde{\uuu}^T\tilde{\uuu}\\
\mathrm{s.t.} & \enspace \WWW\tilde{\uuu}=\Sigmaaa\Gammaaa^{*^{1/2}}\;(\bbb+\ccc)+\Deee,
\end{aligned}
\end{equation}
where $\Gammaaa^*=\mathrm{blkdiag}(\gamma_1^*,...,\gamma_K^*)$ is the input vector of SINR thresholds given by the optimal solution of
\begin{equation}\label{eq:sbw}
\begin{aligned}
\tilde{\uuu}_\mathrm{SB}(P)\!=\!
\mathrm{arg}\underset{\tilde{\uuu},\Gammaaa,\Deee\succeq\mathbf{0}}{\mathrm{max}} & \enspace \underset{k}{\min}\left\{\gamma_k\right\}_{k=1}^K \\
\mathrm{s.t.} & \enspace \WWW\tilde{\uuu}=\Sigmaaa\Gammaaa^{1/2}\;(\bbb+\ccc)+\Deee,\\
& \enspace \tilde{\uuu}^T\tilde{\uuu}\leq P,
\end{aligned}
\end{equation}
until the minimum power solution of \eqref{eq:pmw} converges to $P$. It can be inferred that the power optimization \eqref{eq:pmw} and the max-min SINR \eqref{eq:sbw} are related as
\begin{equation}\label{eq:pmsm} \tilde{\uuu}_\mathrm{PM}(\Gammaaa^*)=\tilde{\uuu}_\mathrm{SB}\Big(\tilde{\uuu}_\mathrm{PM}(\Gammaaa^*)^T\tilde{\uuu}_\mathrm{PM}(\Gammaaa^*)\Big).
\end{equation}

In fact, $\gamma_k$ in \eqref{eq:sbw} manipulates the instantaneous average power of the constellation, from which $\mathcal{D}_{i_k,\text{\tiny DP}}$ is constructed, in order to ensure $\tilde{\uuu}^T\HHH_k^T\HHH_k\tilde{\uuu}\geq\sigma^2\gamma_k$ through the first constraint. This is a conservative way to guarantee that the long-term (e.g., frame-level) SINRs satisfy $\mathbb{E}\{\tilde{\uuu}^T\HHH_k^T\HHH_k\tilde{\uuu}\}/\sigma^2\geq\gamma_k, k=1,...,K$, which is typically desired in conventional multiuser precoding \cite{tb_coor}.
The optimal solution $\gamma_k^*$, however, causes $\HHH_k\tilde{\uuu}$ to lie on $\sqrt{\gamma_k^*}\mathcal{D}_{i_k,\text{\tiny DP}}$, instead of $\mathcal{D}_{i_k,\text{\tiny DP}}$. Since $\gamma_k^*$ is a function of the users' symbol vector $\s$, it varies over symbol time, limiting the applicability of this method to constant envelope modulations. For generic constellations, possibly having points with bounded decision regions, the $k$-th receiver needs to be aware of the value of $\gamma_k^*$ in each symbol period in order to correctly detect $s_k$, which is practically unrealistic. It is important to note that we are not allowed to reformulate \eqref{eq:sbw} by excluding the constraints related to the users $k\notin\mathcal{K}$, as the power optimization \eqref{eq:pmw} needs to take into account all the users' symbols in order to guarantee the given SINR thresholds for all $K$ users.

Assuming identical noise distributions across the receivers, i.e., $\sigma_k^2=\sigma^2, k=1,...,K$, the symbol-level SINR for user $k$ is proportional to the instantaneous received power by the $k$-th receiver at each symbol time. On this account, the DPCIR-based SLP max-min SINR problem can be formulated as
\begin{equation}\label{eq:sb}
\begin{aligned}
\underset{\tilde{\uuu}, \Deee\succeq\mathbf{0}}{\mathrm{maximize}} & \quad\underset{k}{\min}\left\{\tilde{\uuu}^T\HHH_k^T\HHH_k\tilde{\uuu}\right\}_{k\in\mathcal{K}} \\
\mathrm{s.t.} & \quad \WWW \tilde{\uuu}=\sigma\;(\bbb+\ccc)+\Deee,\\
& \quad \tilde{\uuu}^T\tilde{\uuu}\leq P.
\end{aligned}
\end{equation}
By introducing a slack variable $\lambda$, one can recast \eqref{eq:sb} as
\begin{equation}\label{eq:sbt}
\begin{aligned}
\underset{\tilde{\uuu}, \Deee\succeq\mathbf{0}}{\mathrm{maximize}} & \quad \lambda \\
\mathrm{s.t.} & \quad \WWW\tilde{\uuu}=\sigma\;(\bbb+\ccc)+\Deee,\\
& \quad \tilde{\uuu}^T\HHH_k^T\HHH_k\tilde{\uuu}\geq \lambda, k\in\mathcal{K},\\
& \quad \tilde{\uuu}^T\tilde{\uuu}\leq P,
\end{aligned}
\end{equation}
which is not convex due to the third set of constraints. In order to deal with this problem, we use the properties of DPCIRs derived in Section \ref{sec:cir}. According to Remark 1, any point in $\mathcal{D}_{{i_k},\text{\tiny DP}}$ can be uniquely specified by $\Deee_{i_k}=[\delta_{i_k,1},\delta_{i_k,2}]^T\in\mathbb{R}^2_+$ for all $\x_{i_k}\in\bd\chi$. It then follows from Theorem \ref{thm:1} that $\tilde{\uuu}^T\HHH_k^T\HHH_k\tilde{\uuu}=\|\HHH_k\tilde{\uuu}\|^2_2$ is strictly increasing in each element of $\Deee_{i_k}$ for all $k\in\mathcal{K}$, i.e., letting either $\delta_{i_k,1}$ or $\delta_{i_k,2}$ be fixed, $\tilde{\uuu}^T\HHH_k^T\HHH_k\tilde{\uuu}$ is a monotonically increasing function of the other. This suggests that in case the optimal value of one of the elements, e.g., $\delta_{i_k,1}$, is given for any user $k\in\mathcal{K}$, then maximizing $\tilde{\uuu}^T\HHH_k^T\HHH_k\tilde{\uuu}$ is equivalent to maximizing $\delta_{i_k,2}$. In other words, by fixing one of the variables $\delta_{i_k,1}$ or $\delta_{i_k,2}$ for all users $k\in\mathcal{K}$, the problem can be formulated as a convex optimization problem. Assuming $\delta_{i_k,1}, \forall k \in \mathcal{K}$ are fixed, then the convex reformulation of \eqref{eq:sbt} can be written as
\begin{equation}\label{eq:sbtfixed}
\begin{aligned}
\underset{\tilde{\uuu}, \Deee\succeq\mathbf{0}}{\mathrm{maximize}} & \quad \lambda \\
\mathrm{s.t.} & \quad \WWW \tilde{\uuu}=\sigma\;(\bbb+\ccc)+\Deee,\\
& \quad \delta_{i_k,2}\geq \lambda, k\in\mathcal{K},\\
& \quad \tilde{\uuu}^T\tilde{\uuu}\leq P,
\end{aligned}
\end{equation}
where $\delta_{i_k,2}$ is substituted for $\tilde{\uuu}^T\HHH_k^T\HHH_k\tilde{\uuu}$ in \eqref{eq:sbt}. Notice that variables $\Deee_{i_k}, \forall k\notin\mathcal{K}$, which correspond to the constellation points with bounded decision regions, are automatically set to zero by the optimization \eqref{eq:sbtfixed} since the only feasible point satisfying the first constraint for any $k\notin\mathcal{K}$ is $\Deee_{i_k}=\mathbf{0}$.

In theory, achieving the optimum of \eqref{eq:sbt} through \eqref{eq:sbtfixed} requires a complete search over all possible (non-negative) values of $\delta_{i_k,1}, k\in\mathcal{K}$, solving \eqref{eq:sbtfixed} for each choice, and finally picking the maximum among all the candidate solutions. Due to the power limitation induced by $P$, one can bound and discretize the search interval to do an exhaustive search. This reduces the optimization to choose $\delta_{i_k,1}, k\in\mathcal{K}$ from a finite set, but of course leads to a sub-optimal solution. Considering an identical search interval for all users, let $N_\delta$ be the number of discrete values for $\delta_{i_k,1}, k\in\mathcal{K}$, which yields a total number of $N_\delta^{|\mathcal{K}|}$ combinations over all $|\mathcal{K}|$ users. This means that the number of convex problems to be solved every symbol time is of order $N_\delta^{|\mathcal{K}|}$. In general, the gap to the optimal solution depends on $N_\delta$, and also on the bounding accuracy (i.e., whether the search interval includes the optimal value or not). The output of the exhaustive search tends to the optimum of \eqref{eq:sbt} as $N_\delta\rightarrow\infty$, however, the computational complexity grows exponentially with $N_\delta$. Motivated by the very high and impractical complexity of the exhaustive search method, we suggest two more computationally tractable approaches to solve the SLP max-min SINR problem. The proposed alternatives are not equivalent to the original problem \eqref{eq:sbt}, but extensively reduce the computational complexity of the solution method compared to the exhaustive search. In Section \ref{sec:sim}, the loss due to each proposed method will be estimated through simulation results.

\subsubsection{Semidefinite program formulation}

Inspired by the increasing monotonicity of $\tilde{\uuu}^T\HHH_k^T\HHH_k\tilde{\uuu}$ with respect to both elements of $\Deee_{i_k}$ for all $k\in\mathcal{K}$, we propose an alternative way to convert \eqref{eq:sbt} into a convex problem by replacing the non-convex quadratic constraints on $\tilde{\uuu}^T\HHH_k^T\HHH_k\tilde{\uuu}$ with affine constraints on $\Deee_{i_k}$, i.e.,
\begin{equation}\label{eq:sbdu}
\begin{aligned}
\underset{\tilde{\uuu}, \Deee\succeq\mathbf{0}}{\mathrm{maximize}} & \quad \lambda \\
\mathrm{s.t.} & \quad \WWW\tilde{\uuu}=\sigma\;(\bbb+\ccc)+\Deee,\\
& \quad \Deee_{i_k}\succeq \lambda\:\mathbf{1}, k\in\mathcal{K},\\
& \quad \tilde{\uuu}^T\tilde{\uuu}\leq P,
\end{aligned}
\end{equation}
which can be interpreted as jointly maximizing $\delta_{i_k,1}$ and $\delta_{i_k,2}$ over all $k\in\mathcal{K}$. By Schur complement, problem \eqref{eq:sbdu} can be written as
\begin{equation}\label{eq:sbds}
\begin{aligned}
\underset{\tilde{\uuu}, \Deee\succeq\mathbf{0}}{\mathrm{maximize}} & \quad \lambda \\
\mathrm{s.t.} & \quad \WWW\tilde{\uuu}=\Sigmaaa\;(\bbb+\ccc)+\Deee,\\
& \quad \begin{bmatrix}\diag(\Deee_\mathcal{K}) & \mathbf{I}_{2|\mathcal{K}|} \\ \mathbf{I}_{2|\mathcal{K}|} & \lambda\, \mathbf{I}_{2|\mathcal{K}|}\end{bmatrix} \succeq 0,\\
& \quad \begin{bmatrix}1 & \tilde{\uuu}^T \\ \tilde{\uuu} & P \mathbf{I}_{2N}\end{bmatrix}\succeq 0,
\end{aligned}
\end{equation}
where $\Deee_\mathcal{K}\in\mathbb{R}_+^{2|\mathcal{K}|}$ is the vector collecting $\Deee_{i_k}$ for all $k\in\mathcal{K}$, and $\succeq0$ denotes positive semidefinite. Problem \eqref{eq:sbds} is a standard semidefinite program (SDP) and can be solved using known methods. This convex formulation, however, is not expected to achieve the same solution as compared to the original problem \eqref{eq:sbt} since it has a reduced degrees of freedom to maximize the minimum SINR. More precisely, it optimizes $\mathrm{min}\{\delta_{i_k,1},\delta_{i_k,2}\}$ instead of optimizing each of them separately. Nonetheless, the optimal solution of problem \eqref{eq:sbds} can be regarded as a lower bound on the optimum of the SLP max-min SINR. It is also important to note that the SDP \eqref{eq:sbds} is equivalent to the SOCP formulation of SLP SINR balancing proposed for PSK constellations in \cite{slp_chr}. The structure of constraints in both the problems promotes equity in optimizing $\delta_{i_k,1}$ and $\delta_{i_k,2}$ rather than exploiting the entire region to accommodate the received signal; therefore, they achieve optimality if the objective function is defined as the minimum (among $K$ users) instantaneous average power of the receiver's reference constellation. We will clarify the difference with the original max-min SINR problem through the simulation results in Section \ref{sec:sim}.

\subsubsection{Block Coordinate Descent Optimization}

In order to improve the solution of SDP convex formulation \eqref{eq:sbds}, we propose an iterative method based on the block coordinate descent (BCD) algorithm \cite{nonlinear_bert}. The BCD is a family of successive lower-bound maximization methods in which certain approximate version of the objective function is optimized with respect to one block variable at a time, while fixing the rest of the block variables. We denote by $\Deee_{\mathcal{K},1}\in\mathbb{R}_+^{|\mathcal{K}|}$ and $\Deee_{\mathcal{K},2}\in\mathbb{R}_+^{|\mathcal{K}|}$ the vectors (blocks) collecting $\delta_{i_k,1}$ and $\delta_{i_k,2}$ for all $k\in\mathcal{K}$, respectively. The idea behind the BCD-based method is then to successively maximize the worst-user SINR along coordinates $\Deee_{\mathcal{K},1}$ and $\Deee_{\mathcal{K},2}$ until convergence is reached. In more details, defining the monotonically increasing function $f_k: \mathbb{R}_+^2\mapsto\mathbb{R}$ as
\begin{equation}\label{eq:f}
f_k(\delta_{i_k,1},\delta_{i_k,2})=\tilde{\uuu}^T\HHH_k^T\HHH_k\tilde{\uuu}, k\in\mathcal{K},
\end{equation}
the objective function of the SLP max-min SINR can be expressed as
\begin{equation}\label{eq:f2}
\begin{aligned}
g(\Deee_{\mathcal{K},1},\Deee_{\mathcal{K},2}) = \underset{k}{\mathrm{min}}\Big\{f_k(\delta_{i_k,1},\delta_{i_k,2})\Big\}_{k\in\mathcal{K}}.
\end{aligned}
\end{equation}
At a given iteration $n$, each block of variables is updated by using the following objective functions (the constraints are as before)
\begin{equation}\label{eq:f3}
\begin{aligned}
\Deee^*_{\mathcal{K},1|n} = \mathrm{arg}\underset{\Deee_{\mathcal{K},1}}{\mathrm{max}}\quad g(\Deee_{\mathcal{K},1},\Deee^*_{\mathcal{K},2|n-1}),
\end{aligned}
\end{equation}
\begin{equation}\label{eq:f4}
\begin{aligned}
\Deee^*_{\mathcal{K},2|n} = \mathrm{arg}\underset{\Deee_{\mathcal{K},2}}{\mathrm{max}}\quad g(\Deee^*_{\mathcal{K},1|n-1},\Deee_{\mathcal{K},2}),
\end{aligned}
\end{equation}
where $\Deee^*_{\mathcal{K},1|n}$ and  $\Deee^*_{\mathcal{K},2|n}$ denote respectively the optimal solutions obtained at the $n$-th iteration, and $g(\Deee_{\mathcal{K},1},\Deee^*_{\mathcal{K},2|n-1})$ and $g(\Deee^*_{\mathcal{K},1|n-1},\Deee_{\mathcal{K},2})$ are approximations of $g(\Deee_{\mathcal{K},1},\Deee_{\mathcal{K},2})$.  We adopt a cyclic updating rule, i.e., the BCD cyclically solves the two SDPs
\begin{equation}\label{eq:sbccd1}
\begin{aligned}
\underset{\tilde{\uuu}, \Deee\succeq\mathbf{0}}{\mathrm{maximize}} & \quad \lambda \\
\mathrm{s.t.} & \quad \WWW\tilde{\uuu}=\Sigmaaa\;(\bbb+\ccc_\text{\tiny DP})+\Deee,\\
& \quad \begin{bmatrix}\diag(\Deee_{\mathcal{K},1}) & \mathbf{I} \\ \mathbf{I} & \lambda \mathbf{I}\end{bmatrix}\succeq 0,\\
& \quad \begin{bmatrix}1 & \tilde{\uuu}^T \\ \tilde{\uuu} & P \mathbf{I}\end{bmatrix}\succeq 0,
\end{aligned}
\end{equation}
and
\begin{equation}\label{eq:sbccd2}
\begin{aligned}
\underset{\tilde{\uuu}, \Deee\succeq\mathbf{0}}{\mathrm{maximize}} & \quad \lambda \\
\mathrm{s.t.} & \quad \WWW\tilde{\uuu}=\Sigmaaa\;(\bbb+\ccc_\text{\tiny DP})+\Deee,\\
& \quad \begin{bmatrix}\diag(\Deee_{\mathcal{K},2}) & \mathbf{I}\\ \mathbf{I} & \lambda \mathbf{I}\end{bmatrix}\succeq 0,\\
& \quad \begin{bmatrix}1 & \tilde{\uuu}^T \\ \tilde{\uuu} & P \mathbf{I}\end{bmatrix}\succeq 0,
\end{aligned}
\end{equation}
where the dimensions of identity matrices in \eqref{eq:sbccd1} and \eqref{eq:sbccd2} are the same as in \eqref{eq:sbds}.
Each SDP is solved with respect to one of the blocks $\Deee_{\mathcal{K},1}$ or $\Deee_{\mathcal{K},2}$ while the other block is fixed and given by the solution of the previous iteration. The pseudocode of the proposed method is given in Algorithm \ref{alg:1}, where we have arbitrarily initialized $\Deee^*_{\mathcal{K},2}$. For all iterations $n=1,2,...$, we have
\begin{equation}\label{eq:con1}
\begin{aligned}
\Deee^*_{\mathcal{K},1|n-1} \preceq \Deee^*_{\mathcal{K},1|n},\;\;
\Deee^*_{\mathcal{K},2|n-1} \preceq \Deee^*_{\mathcal{K},2|n},
\end{aligned}
\end{equation}
and hence
\begin{equation}\label{eq:con1}
\lambda_{|n-1}^* \leq \lambda_{|n}^*,
\end{equation}
where by $\lambda_{|n}^*$ we denote the optimal solution at the $n$-th iteration.
The sequence $\{\lambda_{|n}^*\}_{n=1,2,...}$ is therefore guaranteed to converge to a stationary
point (local optimum) of the SLP max-min SINR. As we will see in Section \ref{sec:sim}, the BCD algorithm usually converges after a few iterations.

\begin{algorithm}[!tp]
	\caption{Block Coordinate Descent Algorithm to solve the DPCIR-based SLP max-min SINR}
	\label{alg:1}
	\begin{algorithmic}[1]
		\linespread{1.2}
        \State \bf{input:} $\s,\{\h_k\}_{k=1}^K,\Sigmaaa,P, \epsilon$
        \State \bf{initialize:} $n\leftarrow 0, \Deee^*_{\mathcal{K},2|0}\leftarrow \mathbf{0}_{|\mathcal{K}|}$
		\Repeat
		\State $n \leftarrow n+1$
		\If    {$n\; \mathrm{is\;odd}$}
		\State $\Deee^*_{\mathcal{K},2|n}\leftarrow\Deee^*_{\mathcal{K},2|n-1}$
        \State solve \eqref{eq:sbccd1}
		\State return $\lambda_{|n}^*, \Deee^*_{\mathcal{K},1|n}$
		\Else
		\State $\Deee^*_{\mathcal{K},1|n}\leftarrow\Deee^*_{\mathcal{K},1|n-1}$
		\State solve \eqref{eq:sbccd2}
		\State return $\lambda_{|n}^*, \Deee^*_{\mathcal{K},2|n}$
		\EndIf
		\Until $|\lambda_{|n}^*-\lambda_{|n-1}^*|\leq\epsilon$
        \State \bf{output:} $\tilde{\uuu}$
	\end{algorithmic}
\end{algorithm}

\section{Simulation Results}\label{sec:sim}

In this section, we provide some simulation results to validate the analytical discussion in earlier sections and also to evaluate the performance of the proposed SLP design approaches. We compare the results with state of the art. In all the simulations, we consider a downlink multiuser unicast scenario (with equal number of transmit and receiver antennas, i.e., $N=K$) in which intended symbols of all the users are taken from an identical constellation set. We examine the results for three constellations, namely, 8-PSK, optimized 8-ary and 16-QAM; however, we are particularly interested in the optimized 8-ary constellation since it has a generic shape with unequal distances as well as points with both bounded and unbounded Voronoi region. We assume the variance of the noise component at the receiver of each user to be $\sigma_k^2=\sigma^2=1,k=1,...,K$. Furthermore, we assume equal SINR thresholds $\gamma_k=\gamma,k=1,...,K$ when the power minimization is of interest. A quasi-static Rayleigh fading channel is assumed where the complex channel vector $\h_k, k=1,...,K$ is generated following an i.i.d. complex Gaussian distribution with zero-mean and unit variance, with assumption $\mathbb{E}\{\h_k\h_j^H\}=0, \forall j=1,...,K,j\neq k$. As for the BCD algorithm, we set the terminating condition $\epsilon=10^{-3}$ with a maximum number of iterations 100.

For a power-limited downlink scenario with $N=K=4$, the feasibility probability of the DPCIR-based SLP scheme is obtained (based on Proposition \ref{pro:0}) and shown in Fig. \ref{fig:PF_P}. The adopted constellation is the optimized 8-ary. The probabilities are calculated by averaging over all $2^{12}$ possible combinations of the users' symbol vector $\s$, and also averaging over $1000$ randomly generated channel realizations. It can be noticed that for smaller values of $\gamma$, the probability of feasibility grows faster as a function of the available transmit power $P$. A case-specific example could be wireless systems with adaptive coding and modulation (ACM) capability, such as DVB-S2X broadcasting standard \cite{dvbs2x}. In DVB-S2X, the target range of SNR for an 8-ary constellation is typically around $5$-$7$ dB over a linear channel (notice that in SLP, SINR can be interpreted as the received SNR). In such system with a total power budget of at least $130$ dBW, one can say from Fig. \ref{fig:PF_P} that providing all the users with an SINR (SNR) level of $\gamma=5$ dB is guaranteed by $90\%$, and further reduction of transmit power may be possible via the SLP power optimization.
\begin{figure}
	\centering
	\includegraphics[width=.55\columnwidth]{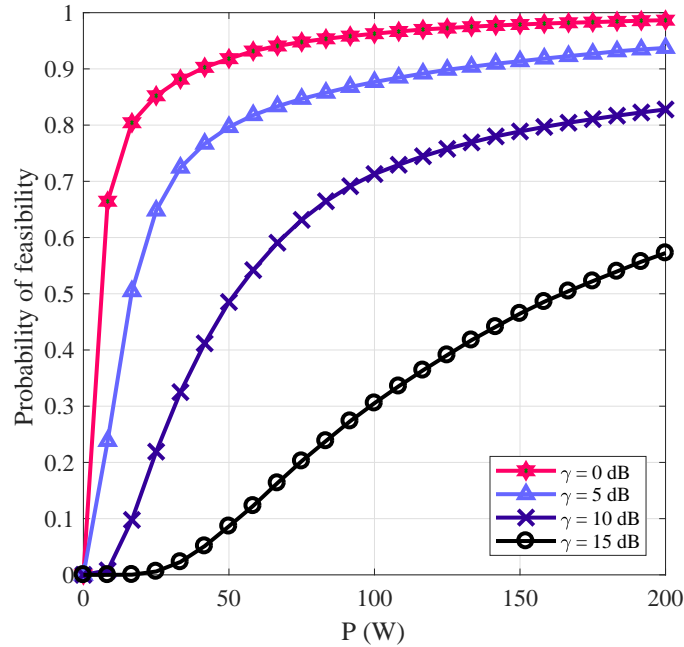}
	\caption{Feasibility probability of SLP as a function of the transmit power budget for different SINR thresholds with $N=K=4$.}
	\label{fig:PF_P}
\end{figure}
\begin{figure}
	\centering
	\includegraphics[width=.55\columnwidth]{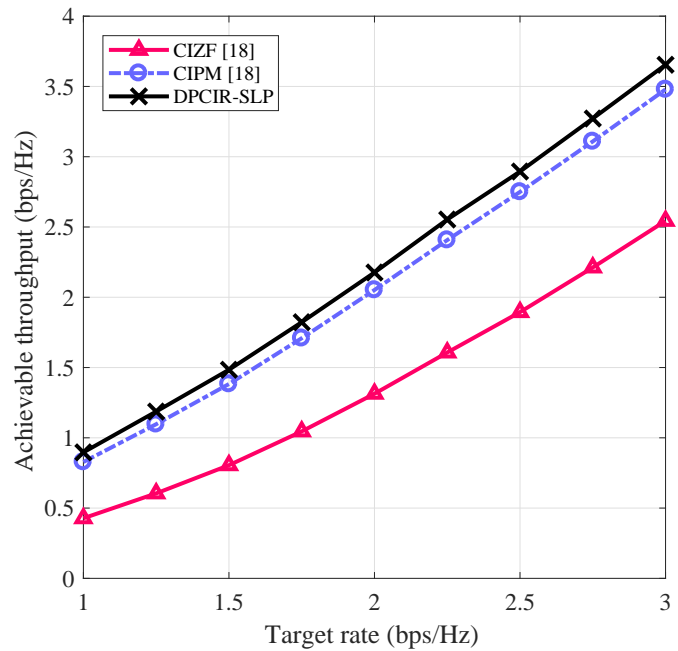}
	\caption{Average per-user achievable throughput versus target rate with $N=K=8$.}
	\label{fig:TH_R}
\end{figure}
\begin{figure}
	\centering
	\includegraphics[width=.55\columnwidth]{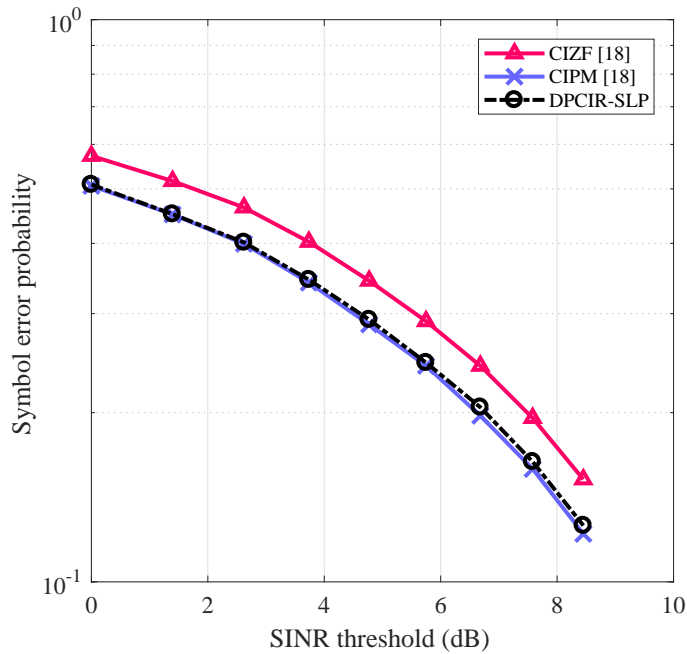}
	\caption{Average symbol error probability versus SINR threshold with $N=K=8$.}
	\label{fig:SER_SINR}
\end{figure}
In Fig. \ref{fig:TH_R}, we plot the average achievable throughput of $K=8$ users under the SLP power minimization scheme as a function of a given target rate $R$, where the target rate is related to the SINR threshold by $R=\log_2{(1+\gamma)}$. The number of BS's transmit antennas is $N=8$ and the optimized 8-ary constellation is employed. The achievable throughput for user $k$ is defined to be equal to
\begin{equation}\label{eq:th}\nonumber
(1-\mathrm{SER}_k)\log_2{\left(1+\mathbb{E}\left\{\|\h_k\uuu\|_2^2\right\}\right)}
\end{equation}
where $\mathrm{SER}_k$ is the symbol error rate of the $k$-th user, and the expectation is taken over each frame. In addition to the DPCIR-based SLP design, the results are obtained for two other SLP approaches, namely, constructive interference zero-forcing (CIZF) and constructive interference power minimization (CIPM) \cite{slp_con}. The proposed DPCIR based scheme outperforms both CIZF and CIPM. 
It can also be observed that both the DPCIR-based and the CIPM symbol-level precoders provide higher achievable throughputs than the given target rate. 
Moreover, under the same scenario, the average symbol error probability over all $K$ users is depicted versus SINR threshold in Fig. \ref{fig:SER_SINR}. As it can be seen, defining the CIR constraints of the SLP power optimization to be distance preserving causes a very slight difference in the average SER compared to the CIPM approach in which the phase of the noise-free received signal is pushed to agree with that of the original constellation point. Overall, with respect to Fig. \ref{fig:TH_R}, the DPCIR-based SLP shows a better performance than the CIPM in terms of the achievable throughput \eqref{eq:th} which takes into account both the shape of the CIR and the resulting SER.

Figure \ref{fig:SCATTER} shows the scatter plot of $K\times10^3$ noise-free received signals in a scenario with $K=N=8$ and $\gamma=15$ dBW, where all the transmitted symbols are drawn from 8-PSK constellation and mapped to transmit antennas via a DPCIR-based SLP max-min SINR precoder. This figure confirms the discussion in Section \ref{sec:slp} regarding the relative location of the noise-free received signal on its corresponding DPCIR. It can be seen from Fig. \ref{fig:SCATTER} that the density of signals resulted from the BCD algorithm is higher in areas closer to the boundaries of DPCIRs, while those signals from the convex approximation are distributed around the bisector (with the majority lying exactly on the bisector). This is a consequence of maximizing (the minimum of) the two parameters $\delta_{i_k,1}$ and $\delta_{i_k,2}$ in \eqref{eq:sbds} which disregards one degree of freedom in optimization for each user $k\in\mathcal{K}$. On the other hand, as it can be seen in Fig. \ref{fig:SCATTER}, the results obtained from the BCD algorithm are biased towards one of the edges in each DPCIR, depending on the initialization step (i.e., $\delta_{{i_k},1}$ and $\delta_{{i_k},2}$). The exact same trend can be observed in Fig. \ref{fig:SCATTER} for the output of the SOCP solving SLP max-min SINR in \cite{slp_chr}. As mentioned in Section \ref{sec:slp}, as far as optimizing the instantaneous average power of the constellation is of concern, the proposed convex approximation \eqref{eq:sbds} and the SOCP \cite{slp_chr} lead to the optimal max-min SINR solution.
\begin{figure}
	\centering
	\includegraphics[width=.55\columnwidth]{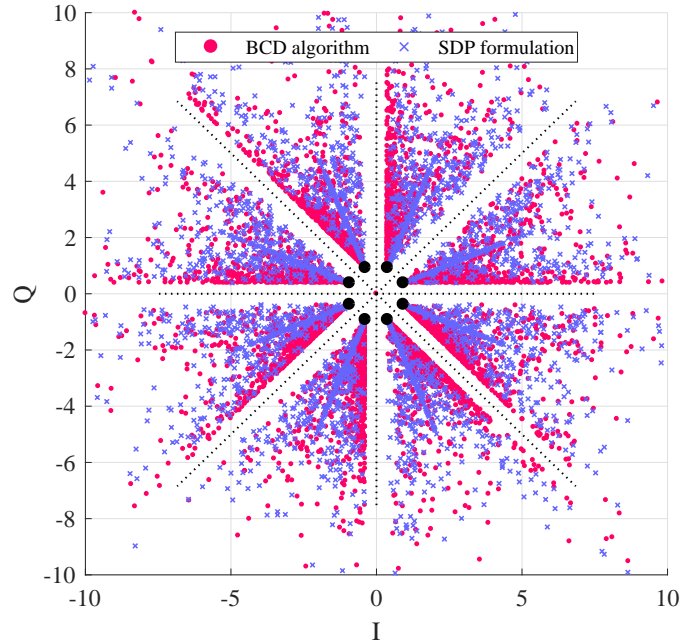}
	\caption{Scatter plot of the noise-free received signals taken form 8-PSK constellation with $N=K=8$ and $\gamma=15$ dBW. The black points and the dashed lines represent the constellation points and their corresponding Voronoi regions, respectively.}
	\label{fig:SCATTER}
\end{figure}
Figures \ref{fig:SINR_P_PSK}-\ref{fig:SINR_P_QAM} plot the optimized worst-user SINR obtained via different SLP SINR balancing approaches for three constellations 8-PSK, optimized 8-ary and 16-QAM, respectively. We also compare the results with those of the maximal fairness zero-forcing precoder \cite{tb_zf} and the bisection algorithm \cite{slp_con}. The method based on exhaustive search is used as a basis for comparison. We separately take $N_{\delta}=5$ and $N_{\delta}=7$ points to search over the interval $[0,2.5]$. The SDP formulation, while being always superior to the maximal fairness ZF precoding by at least $1$ dB, is a lower bound on the optimal SLP max-min SINR solution. The BCD algorithm, on the other hand, provides gains up to $2$ dB with respect to the convex approximation for optimized 8-ary constellation. Furthermore, Fig. \ref{fig:SINR_P_OPT} shows that this successive method could achieve even better solutions than the exhaustive search with $N_{\delta}=7$ when the optimized 8-ary constellation is employed.
\begin{figure}
	\centering
	\includegraphics[width=.55\columnwidth]{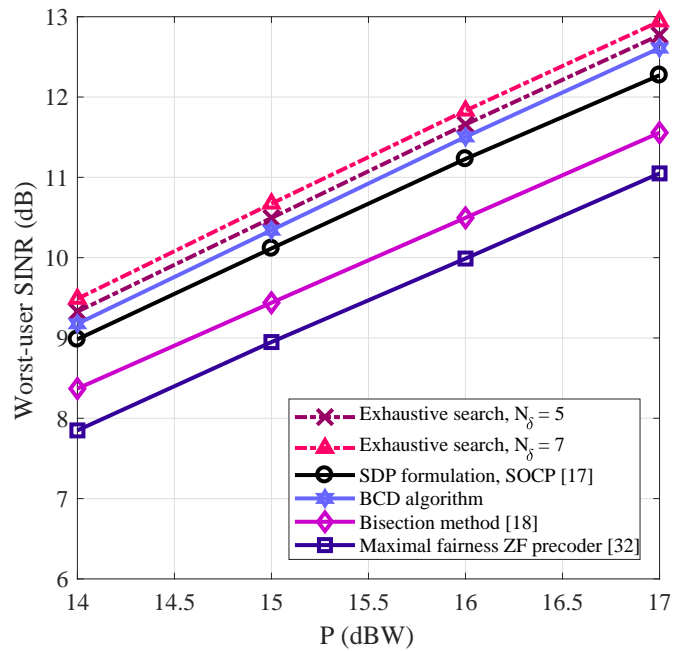}
	\caption{Minimum SINR among $K=4$ users versus total power constraint for 8-PSK constellation.}
	\label{fig:SINR_P_PSK}
\end{figure}
\begin{figure}
	\centering
	\includegraphics[width=.55\columnwidth]{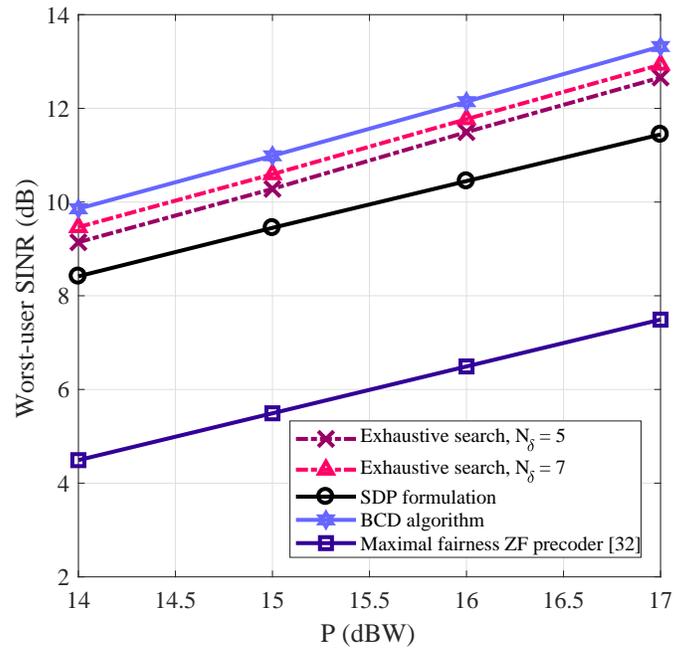}
	\caption{Minimum SINR among $K=4$ users versus total power constraint for the optimized 8-ary constellation.}
	\label{fig:SINR_P_OPT}
\end{figure}
\begin{figure}
	\centering
	\includegraphics[width=.55\columnwidth]{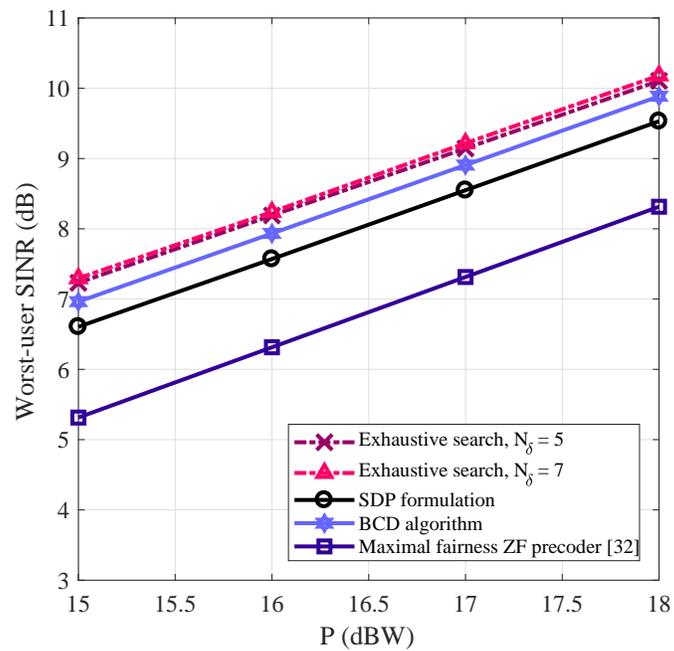}
	\caption{Minimum SINR among $K=4$ users versus total power constraint for 16-QAM constellation.}
	\label{fig:SINR_P_QAM}
\end{figure}
In Fig. \ref{fig:SINR_NK_P}, the optimized worst-user SINR is plotted as a function of the system dimension, where the users' symbols are taken from the optimized 8-ary constellation. As expected, a lower minimum SINR is achieved with increasing the system dimension; however, SINR drops more gradually with respect to the system dimension for larger power budgets. 
\begin{figure}
	\centering
	\includegraphics[width=.55\columnwidth]{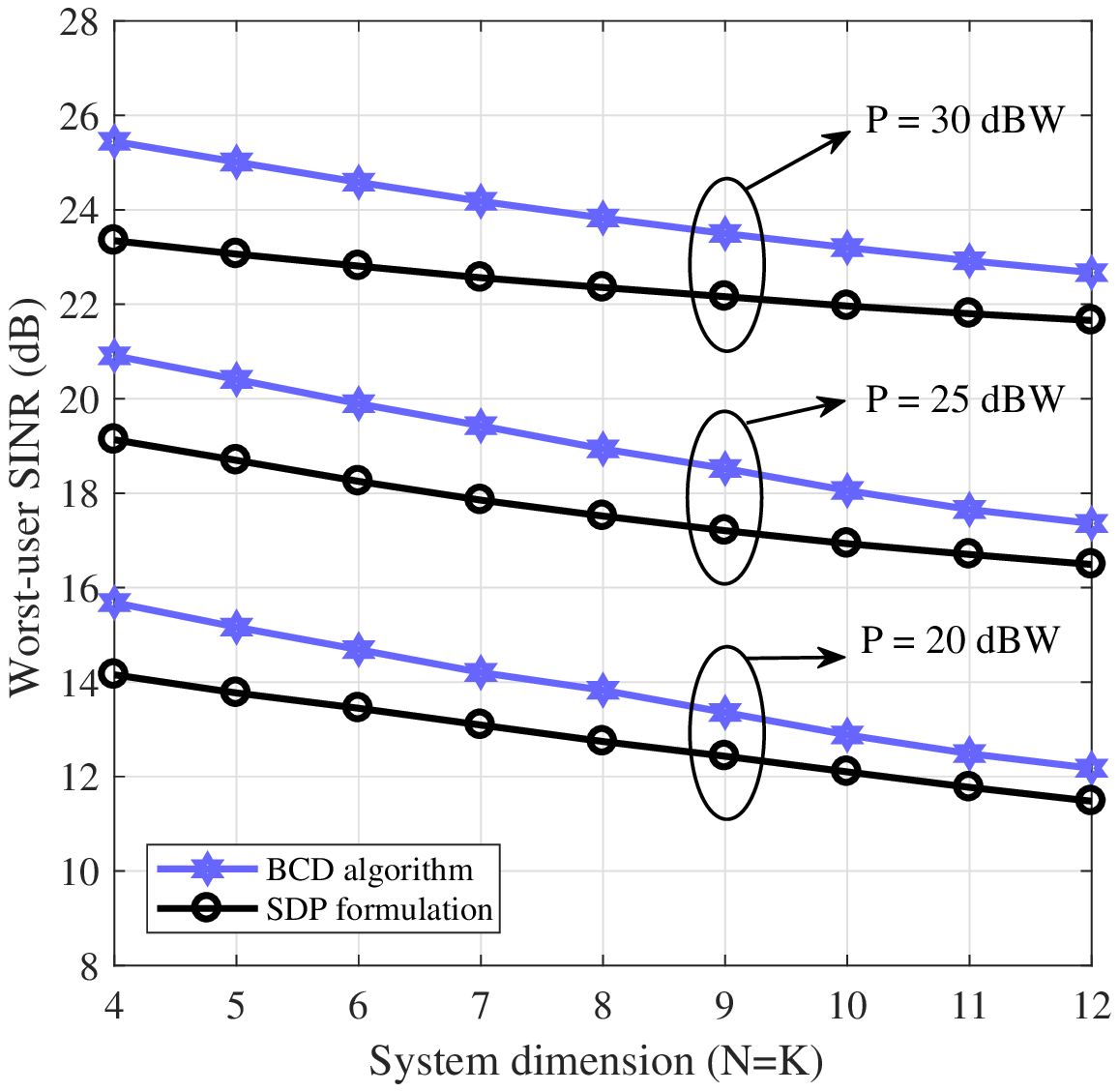}
	\caption{The worst-user SINR as a function of system dimension for different values of total power budget.}
	\label{fig:SINR_NK_P}
\end{figure}
In Fig. \ref{fig:ITR_NK_P}, we present the convergence rate of the BCD algorithm versus the system dimension for different power budgets for 8-PSK and the optimized 8-ary constellation. Here, the convergence rate is expressed in terms of the average number of iterations until the terminating condition is met. It can be seen that the algorithm solving the SLP max-min SINR for $N=K=4$ converges after a few iterations with an average of up to 6 iterations for $P=30$ dB, where each iteration consists a single convex optimization problem. Figure \ref{fig:ITR_NK_P} also shows a slightly slower convergence behavior for higher values of $P$ which is due to a larger feasible region.

\noindent{\it{\bf Complexity comparison:}} In a scenario similar to that of Fig. \ref{fig:SINR_P_PSK} and Fig. \ref{fig:SINR_P_OPT}, i.e., with $N=K=4$ and assuming $N_\delta=5$ and $N_\delta=7$, the exhaustive search method respectively solves $5^4$ and $7^4$ convex optimization problems at each symbol slot and picks the best solution. The SDP formulation, on the other hand, always solves a single convex optimization problem per symbol time. According to Fig. \ref{fig:ITR_NK_P}, the results for the BCD algorithm are obtained through 4 iterations (optimized 8-ary) and 6-8 iterations (8-PSK), on average, where each iteration involves solving a single SDP. The BCD algorithm, though having higher complexity compared to the SDP formulation, shows $1.5$-$2$ dB (optimized 8-ary) and $0.2$-$0.4$ dB (8-PSK) gain in the worst-user SINR (see Fig. \ref{fig:SINR_P_PSK} and \ref{fig:SINR_P_OPT}). Furthermore, the results obtained from the exhaustive search are almost comparable to those of the BCD algorithm, whereas the latter method is far less computationally complex. Therefore, the BCD algorithm provides a very good complexity-performance tradeoff. 
\begin{figure}
	\centering
	\includegraphics[width=.55\columnwidth]{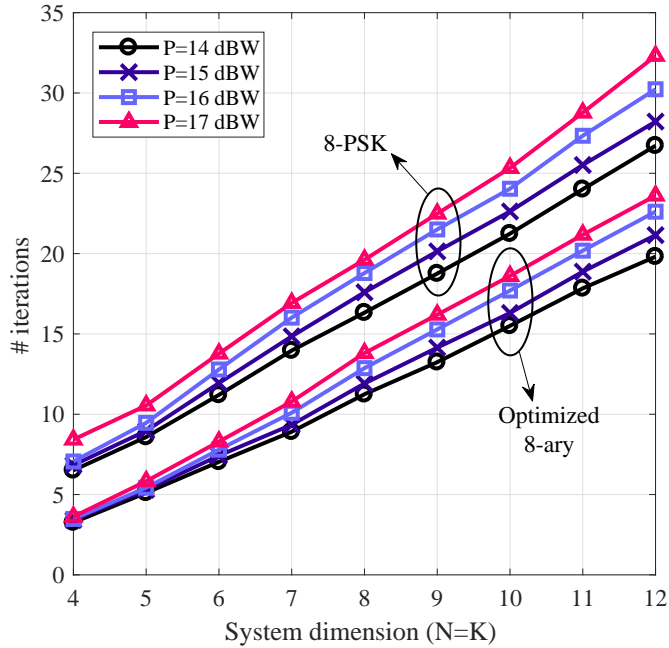}
	\caption{Number of iterations until convergence of the BCD algorithm as a function of system dimension for different values of total power budget.}
	\label{fig:ITR_NK_P}
\end{figure}

\section{Conclusion}\label{sec:conc}
In this paper, we addressed two common precoding design problems in a downlink multiuser channel, namely, power optimization and SINR balancing, on a symbol-level basis. CIRs are the key to define the SLP problem as they determine the constraints yielding constructive interference at the receiver. We considered a general category of CIRs named as DPCIR. Full characterization of DPCIRs for a generic constellation was done which led to extracting some properties for these regions. Using a systematic description for DPCIRs, we formulated and discussed the SLP optimization problems. The SINR-constrained SLP power minimization was investigated in a realistic scenario and a feasibility condition was derived. We also expressed this problem in a simplified equivalent form. For the more challenging and generally non-convex problem of SLP SINR balancing under max-min fairness criterion, the properties of DPCIRs helped us to reformulate the problem in a convex form. Subsequently, two alternative approaches were proposed, namely, SDP formulation and BCD optimization. Finally, we provided a detailed comparison of complexity for the proposed methods.  

\section*{Acknowledgment}
The authors are supported by the Luxembourg National Research Fund (FNR) under CORE Junior project: C16/IS/11332341 Enhanced Signal Space opTImization for satellite comMunication Systems (ESSTIMS).

\appendices
\section{Proof of Lemma \ref{lem:2}}\label{app:lem2}
The intersection of finitely many closed halfspaces is an unbounded polyhedron if and only if the outward normals to the associated boundary hyperplanes lie on a single closed halfspace \cite[p. 20, Theorem 4]{convex_poly}. Accordingly, for any $\x_i\in\chi$ with unbounded $\mathcal{D}_{i,\text{\tiny ML}}$, all the outward normal vectors $-\aaa_{i,j}, j\in\mathcal{J}_i$ lie on a single halfspace. Since the polyhedron $\mathcal{D}_{i,\text{\tiny DP}}$ has the same set of outward normals $-\aaa_{i,j},j\in\mathcal{J}_i$, it is also unbounded. An unbounded polyhedron is uniquely determined from its vertices and the directions of its infinite edges \cite[p. 31, Theorem 4]{convex_poly}. Furthermore, it is straightforward to check that $\x_i$ is the unique solution of $\A_i \x=\bbb_i+\ccc_{i,\text{\tiny DP}}$, i.e., all the contributing hyperplanes have a common intersection point $\x_i$. This means that $\mathcal{D}_{i,\text{\tiny DP}}$, which is given by the solution set of $\A_i \x\succeq\bbb_i+\ccc_{i,\text{\tiny DP}}$, has a single vertex at $\x_i$ and two infinite edges, i.e., a polyhedral angle. In addition, since any two neighboring points share a common Voronoi edge, the two infinite edges of $\mathcal{D}_{i,\text{\tiny DP}}$ correspond to the two neighboring points of $\x_i$ on $\bd\chi$ (i.e., $\mathcal{S}_i\cap\bd\chi$) with unbounded Voronoi regions. Each infinite edge of $\mathcal{D}_{i,\text{\tiny DP}}$ is then parallel to a hyperplane with normal vector $\aaa_{i,j}=\x_i-\x_j$, where $\x_j\in\mathcal{S}_i\cap\bd\chi$; therefore it is perpendicular to $\x_i-\x_j$. This completes the proof.

\section{Proof of Lemma \ref{lem:3}}\label{app:lem3}
In order to prove this lemma, we first state a well-known property of convex sets.

\begin{property}\label{pro:1}
	$\vvv_o$ is the minimum distance vector from the origin to the convex set $\mathcal{V}$ iff for any vector $\vvv\in\mathcal{V}$ we have $\vvv_o^T\vvv\geq\vvv_o^T\vvv_o$, with equality for $\vvv$ lying on the hyperplane orthogonal to $\vvv_o$ \cite[p. 69, Theorem 1]{opt_vec}.
\end{property}

For any $\x_i\in\interior\chi$, Lemma \ref{lem:3} holds straightforwardly as $\mathcal{D}_{i,\text{\tiny DP}}=\x_i$. Therefore, in what follows we only focus on the constellation points belonging to $\bd\chi$.

\emph{Sufficiency}: Having $\mathbf{0}\in\conv\chi$, let further assume that $\mathbf{0}\in\chi$. This assumption, as mentioned earlier in section \ref{sec:cir}, does not have any impact on $\mathcal{D}_{i,\text{\tiny DP}}$ for any $\x_i\in\bd\chi$, regardless of whether $\mathbf{0}\in\bd\chi$ or $\mathbf{0}\in\interior\chi$. By substituting $\x_j=\mathbf{0}$ in \eqref{eq:pro0cons}, for all $\x_i\in\chi$ we have $\|\x\|\geq\|\x_i\|,\forall\x\in\mathcal{D}_{i,\text{\tiny DP}}$. This completes the proof of sufficiency.

\emph{Necessity:} By contradiction, if $\mathbf{0}\notin\conv\chi$, let assume a new constellation set $\tilde{\chi}$ having all the points of $\chi$ including the origin, i.e., $\tilde{\chi}=\chi\cup\{\mathbf{0}\}$, hence $\conv\chi\subset\conv\tilde{\chi}$. Clearly, $\mathbf{0}\in\bd\tilde{\chi}$ and according to Lemma \ref{lem:2}, there always exist exactly two constellation points on $\bd\tilde{\chi}$ that $\mathbf{0}$ contributes to their DPCIRs. Suppose $\x_l$ be one of these points with $\mathcal{D}_{l,\text{\tiny DP}}$ and $\tilde{\mathcal{D}}_{l,\text{\tiny DP}}$ denoting its associated DPCIR relative to $\chi$ and $\tilde{\chi}$, repectively. We denote by $\tilde{\mathcal{S}}_l$ the set of neighboring points of $\x_l$ in $\tilde{\chi}$. Let $\mathcal{H}_{l,o}=\left\{\x\mid\x\in\mathbb{R}^2, \x_l^T\x\geq \x_l^T\x_l\right\}$ be the distance preserving halfspace from $\mathbf{0}$ to $\x_l$. Since $\mathbf{0}\in\tilde{\mathcal{S}}_l$, we have $\tilde{\mathcal{D}}_{l,\text{\tiny DP}}=\mathcal{H}_{l,o}\cap\mathcal{D}_{l,\text{\tiny DP}}\neq\mathcal{D}_{l,\text{\tiny DP}}$, i.e., the halfspace $\mathcal{H}_{l,o}$ does not contain $\mathcal{D}_{l,\text{\tiny DP}}$. Hence, $\left\{\x\mid\x\in\mathbb{R}^2, \x_l^T\x=\x_l^T\x_l\right\}$ is not a supporting hyperplane for $\mathcal{D}_{l,\text{\tiny DP}}$ at $\x_l$ \cite[p. 51]{convex_boyd}. This implies that there exist some $\x\in\mathcal{D}_{l,\text{\tiny DP}}$ for which $\x_l^T\x<\x_l^T\x_l$. According to Property \ref{pro:1} (which gives a necessary and sufficient condition), $\x_l$ is not the minimum distance vector from the origin in $\mathcal{D}_{l,\text{\tiny DP}}$. Consequently, $\|\x\|\geq\|\x_l\|$ does not hold for some $\x\in\mathcal{D}_{l,\text{\tiny DP}}$ which contradicts $\|\x\|\geq\|\x_l\|,\forall\x\in\mathcal{D}_{l,\text{\tiny DP}}$.

\section{Proof of Theorem \ref{thm:1}}\label{app:thm}
In order to prove this theorem we need the following lemma.
\begin{lemma}\label{lem:5}
If $\mathbf{0}\notin\conv\chi$, there exists at least one constellation point $\x_l\in\chi$ for which for any $\x\in\mathcal{D}_{l,\text{\tiny DP}}$, $\mathbf{0}\notin\conv\tilde{\chi}_{\x_l,\x}$, where $\tilde{\chi}_{\x_l,\x}=\chi\cup\{\x\}$.
\end{lemma}

\begin{proof}
	If $\mathbf{0}\notin\conv\chi$, for any $\x_i\in\chi$ and any $\x\in\mathcal{D}_{i,\text{\tiny DP}}$ with $\tilde{\chi}_{\x_i,\x}=\chi\cup\{\x\}$, let define $\mathcal{C}_i=\underset{\x\in\mathcal{D}_{i,\text{\tiny DP}}}\bigcup\conv\tilde{\chi}_{\x_i,\x}$. Having $\conv\chi\subseteq\conv\tilde{\chi}_{\x_i,\x}$, it follows from the definition of convex hull that $\conv\chi=\underset{\x_i\in\chi}\bigcap\mathcal{C}_i$. If $\mathbf{0}\in\mathcal{C}_i, \forall\x_i\in\chi$, then $\mathbf{0}\in\conv\chi$ which contradicts our assumption. Hence there must exist at least one constellation point, say $\x_l$, for which $\mathcal{C}_l$ and therefore none of $\conv\tilde{\chi}_{\x_l,\x}, \forall\x\in\mathcal{D}_{l,\text{\tiny DP}}$ contains the origin, as required.

\end{proof}
Now, we can complete the proof of Theorem \ref{thm:1} as follows.

\emph{Sufficiency:}
Suppose $\mathbf{0}\in\conv\chi$. Assuming a constellation point $\x_i\in\chi$ and its DPCIR $\mathcal{D}_{i,\text{\tiny DP}}$, let $\yyy_1$ and $\yyy_2$ be two points in $\mathcal{D}_{i,\text{\tiny DP}}$ such that $\A_i \yyy_1=\bbb_i+\ccc_{i,\text{\tiny DP}}+\Deee_{i,1}$ and $\A_i \yyy_2=\bbb_i+\ccc_{i,\text{\tiny DP}}+\Deee_{i,2}$ with $\Deee_{i,1},\Deee_{i,2}\in\mathbb{R}^{M_i}_+$ and $\Deee_{i,1}\prec\Deee_{i,2}$. Let consider a new constellation $\tilde{\chi}=\chi\cup\{\yyy_1\}$. It is clear that $\conv\chi\subseteq\conv\tilde{\chi}$, and therefore $\mathbf{0}\in\conv\tilde{\chi}$. The DPCIR of $\yyy_1$ can be described as $\mathcal{D}_{\yyy_1,\text{\tiny DP}}=\big\{\x\mid\x\in\mathbb{R}^2, \A_i \x=\bbb_i+\ccc_{i,\text{\tiny DP}}+\Deee_{i,1}+\Deee_1, \Deee_1\in\mathbb{R}^{M_i}_+ \big\}$. Let $\bar{\Deee}={\Deee}_{i,2}-{\Deee}_{i,1}$, then $\A_i \yyy_2=\bbb_i+\ccc_{i,\text{\tiny DP}}+\Deee_{i,1}+\bar{\Deee}, \bar{\Deee}\in\mathbb{R}^{M_i}_{++}$, which means that $\yyy_2\in\mathcal{D}_{\yyy_1,\text{\tiny DP}}$. As a consequence, from Lemma \ref{lem:3}, we have $\|\yyy_1\|<\|\yyy_2\|$ and the proof of sufficiency is complete.

\emph{Necessity:} 
By contradiction, suppose $\mathbf{0}\notin\conv\chi$. Then, based on Lemma \ref{lem:5}, there exists a constellation point $\x_l$ for which $\mathbf{0}\notin\conv\tilde{\chi}_{\x_l,\x}, \forall\x\in\mathcal{D}_{l,\text{\tiny DP}}$. Let $\yyy_1\in\mathcal{D}_{l,\text{\tiny DP}}$, then $\A_l \yyy_1=\bbb_l+\ccc_{l,\text{\tiny DP}}+\Deee_{l,1}$ with $\Deee_{l,1}\in\mathbb{R}^{M_l}_+$. The DPCIR associated with $\yyy_1$ can be expressed as $\mathcal{D}_{\yyy_1,\text{\tiny DP}}=\big\{\x\mid\x\in\mathbb{R}^2, \A_l \x=\bbb_l+\ccc_{l,\text{\tiny DP}}+\Deee_{l,1}+\Deee_1, \Deee_1\in\mathbb{R}^{M_l}_+ \big\}$. Since $\mathbf{0}\notin\conv\tilde{\chi}_{\x_l,\yyy_1}$, it follows from Lemma \ref{lem:3} and Property \ref{pro:1} that there exists $\yyy_2\in\mathcal{D}_{\yyy_1,\text{\tiny DP}}$ such that $\A_l \yyy_2=\bbb_l+\ccc_{l,\text{\tiny DP}}+\Deee_{l,1}+\bar{\Deee}, \bar{\Deee}\in\mathbb{R}^{M_l}_{++}$, for which $\|\yyy_2\|<\|\yyy_1\|$. But ${\Deee}_{l,1}+\bar{\Deee}={\Deee}_{l,2}$ yields ${\Deee}_{l,2}\succ{\Deee}_{l,1}$ which is a contradiction. This completes the proof.

\bibliographystyle{IEEEtran}

\end{document}